\documentclass[11pt,a4paper]{article}
\usepackage{fullpage}

\usepackage{amsmath}
\usepackage{amsthm}
\usepackage{amssymb}
\usepackage{graphicx}
\usepackage{bm}
\usepackage{amsfonts}
\usepackage{latexsym}
\usepackage{pdfsync}
\usepackage{subfigure}
\usepackage{colordvi}
\usepackage{color}
\usepackage{appendix}
\usepackage{algorithm}
\usepackage{algorithmic}
 \usepackage{multirow}
  \usepackage{tabularx} 
 \usepackage[english]{babel}
 \usepackage{array}
\usepackage{enumitem}
\usepackage{tikz}
\usetikzlibrary{plotmarks}

\begin{document}
\newcommand{\commentout}[1]{}

\newcommand{\nwc}{\newcommand}
\nwc{\ba}{\begin{array}}
\nwc{\bal}{\begin{align}}
\nwc{\bea}{\begin{eqnarray}}
\nwc{\beq}{\begin{eqnarray}}
\nwc{\bean}{\begin{eqnarray*}}
\nwc{\beqn}{\begin{eqnarray*}}
\nwc{\beqast}{\begin{eqnarray*}}

\nwc{\ea}{\end{array}}
\nwc{\eal}{\end{align}}
\nwc{\eea}{\end{eqnarray}}
\nwc{\eeq}{\end{eqnarray}}
\nwc{\eean}{\end{eqnarray*}}
\nwc{\eeqn}{\end{eqnarray*}}
\nwc{\eeqast}{\end{eqnarray*}}

\nwc{\ep}{\varepsilon}
\nwc{\ept}{\epsilon}

\nwc{\calS}{\mathcal{S}}
\nwc{\calT}{\mathcal{T}}
\nwc{\calL}{\mathcal{L}}
\nwc{\calJ}{\mathcal{J}}
\nwc{\calR}{\mathcal{R}}
\nwc{\calE}{\mathcal{E}}
\nwc{\calN}{\mathcal{N}}
\nwc{\calG}{\mathcal{G}}

\nwc{\tcalL}{\tilde\calL}

\nwc{\nwt}{\newtheorem}

\nwt{cor}{Corollary}
\nwt{proposition}{Proposition}
\nwt{corollary}{Corollary}
\nwt{theorem}{Theorem}
\nwt{summary}{Summary}
\nwt{lemma}{Lemma}
\nwt{remark}{Remark}
\nwt{definition}{Definition}
\nwc{\nn}{\nonumber}

\nwc{\diag}{{\rm diag}}

\nwc{\RR}{\mathbb{R}}
\nwc{\CC}{\mathbb{C}}
\nwc{\ZZ}{\mathbb{Z}}
\nwc{\EE}{\mathbb{E}}
\nwc{\PP}{\mathbb{P}}
\nwc{\lam}{\lambda}

\newcommand{\om}{\omega}
\newcommand{\rank}{\text{rank}}
\newcommand{\supp}{\mathcal{S}}

\nwc{\talpha}{\tilde\alpha}
\nwc{\tbeta}{\tilde\beta}
\nwc{\tom}{\tilde\om}
\nwc{\tx}{\tilde{x}}
\nwc{\tg}{\tilde g}
\nwc{\tf}{\tilde f}
\nwc{\widehatom}{\widehat\Gamma}
\nwc{\halpha}{\widehat\alpha}
\nwc{\hbeta}{\widehat\beta}
\nwc{\hsigma}{\widehat\sigma}
\nwc{\hg}{\widehat{g}}
\nwc{\hf}{\widehat{f}}
\nwc{\hR}{\widehat{R}}
\nwc{\hb}{\widehat b}
\nwc{\hF}{\widehat{F}}
\nwc{\hn}{\widehat{n}}
\nwc{\barg}{\bar{g}}
\nwc{\barf}{\bar{f}}
\nwc{\barG}{\bar{G}}
\nwc{\bg}{\mathbf{g}}
\nwc{\bff}{\mathbf{f}}
\nwc{\bh}{\mathbf{h}}
\nwc{\bz}{\mathbf{z}}
\nwc{\bp}{\mathbf{p}}

\nwc{\DR}{{\rm DR}}

\nwc{\adj}{\text{adj}}
\nwc{\trace}{\text{trace}}

\nwc{\wtR}{\widetilde{R}}

\nwc{\smax}{{\sigma_{\max}}}
\nwc{\smaxsq}{{\sigma^2_{\max}}}
\nwc{\smin}{{\sigma_{\min}}}
\nwc{\sminsq}{{\sigma^2_{\min}}}
\nwc{\gammamax}{\gamma_{\max}}
\nwc{\gammamin}{\gamma_{\min}}
\nwc{\alphamax}{\alpha_{\max}}
\nwc{\alphamin}{\alpha_{\min}}
\def\onegroup{\underbrace{0 \hskip\arraycolsep\cdots \hskip\arraycolsep 0}_{k-2}}

\title{Sensor Calibration for Off-the-Grid Spectral Estimation}
\author{Yonina C. Eldar\thanks{Department of EE Technion, Israel Institute of Technology, Haifa. Email: yonina@ee.technion.ac.il}
\and 
Wenjing Liao \thanks{School of Mathematics, Georgia Institute of Technology. Email: wenjing.liao@math.gatech.edu}
\and
Sui Tang \thanks{Department of Mathematics, Johns Hopkins University. Email: stang@math.jhu.edu}
}

\maketitle

\begin{abstract}
This paper studies sensor calibration in spectral estimation where the true frequencies are located on a continuous domain. We consider a uniform array of sensors that collects measurements whose spectrum is composed of a finite number of frequencies, where each sensor has an unknown calibration parameter. Our goal is to recover the spectrum and the calibration parameters simultaneously from multiple snapshots of the measurements. 
In the noiseless case with an infinite number of snapshots, we prove uniqueness of this problem up to certain trivial, inevitable ambiguities based on an algebraic method, as long as there are more sensors than frequencies. We then analyze the sensitivity of this algebraic technique with respect to the number of snapshots and noise. 

We next propose an optimization approach that makes full use of the measurements by minimizing a non-convex objective which is non-negative and continuously differentiable over all calibration parameters and Toeplitz matrices. We prove that, in the case of infinite snapshots and noiseless measurements, the objective vanishes only at equivalent solutions to the true calibration parameters and the measurement covariance matrix. The objective is minimized using Wirtinger gradient descent which is proven to converge to a critical point.
We show empirically that this critical point provides a good approximation of the true calibration parameters and the underlying frequencies.


\end{abstract}


\section{Introduction}


High-performance systems in signal processing often require precise calibration of sensors. However, such advanced sensors can be very expensive and difficult to build. It is, therefore, beneficial to use the measurements themselves to  adjust calibration parameters and perform signal recovery simultaneously. 
We treat sensor calibration in spectral estimation where the frequencies of interest are located on a continuous domain.

A uniform array of $N$ sensors collects measurements whose spectrum is composed of $s$ spikes located at $\calS:=\{\om_j\in [0,1)\}_{j=1}^s$with amplitudes $x(t) := \{x_j(t) \}_{j=1}^s\in \CC^s$ at time $t$. Each sensor has an unknown calibration parameter $g_n,\ n=0,\ldots,N-1$. The measurement vector at the array output can then be written as
\beq
\label{eqprob1}
y_e(t) = GAx(t)+e(t),
\eeq
where $y_e(t) = \{y_{e,n}(t)\}_{n=0}^{N-1} \in \CC^N$ and $y_{e,n}(t)$ is the noisy measurement collected by the $n$-th sensor at time $t$, $G = \diag(g) \in \CC^{N \times N}$ with $g = \{g_n\}_{n=0}^{N-1} \in \CC^N$ is the calibration matrix, $e(t) \in \CC^N$ is the noise vector at time $t$, and  $A \in \CC^{N \times s}$ is the sensing matrix with elements
\beq
\label{eqA}
A_{n, j} =  \frac{1}{\sqrt N}{e^{2\pi i n \om_j }}
\eeq
where $n=0,\ldots,N-1$ and $j=1,\ldots,s$.
Our goal is to recover the spectrum $\calS$ and the calibration parameters $g$ simultaneously from noisy measurements $\{y_e(t), t \in\Gamma \}$ where $L:=\#\Gamma$ is the number of snapshots. 

Spectral estimation modeled by \eqref{eqprob1} is a fundamental problem in imaging and signal processing. It is widely used in speech analysis, direction of arrival (DOA) estimation, array imaging and remote sensing. For example, in array imaging \cite{FannjiangYan,ArrayImaging,SchD,Sch}, we assume there are $s$ sources located at $\{\om_j\}_{j=1}^s$ with amplitudes $\{x_j(t)\}_{j=1}^s$, and use a uniform array of sensors to collect measurements $\{y_e(t), \ t \in \Gamma\}$. Our goal is to recover the source locations and the sensor calibration parameters simultaneously from the measurements. 

\commentout{
\begin{figure}
\centering
\commentout{
\begin{tikzpicture}[scale = 0.5]
  \draw[thick] (10,-5) -- (10,5);
   \draw [fill=blue,blue, ultra thick] (10,4.2) circle [radius = 0.2];
   \draw [fill=blue,blue, ultra thick] (10,2.2) circle [radius = 0.2];
   \draw [fill=blue,blue, ultra thick] (10,1) circle [radius = 0.2];
   \draw [fill=blue,blue, ultra thick] (10,0) circle [radius = 0.2];
   \draw [fill=blue,blue, ultra thick] (10,-1.5) circle [radius = 0.2];
   \draw [fill=blue,blue, ultra thick] (10,-2.8) circle [radius = 0.2];
   \draw [fill=blue,blue, ultra thick] (10,-4) circle [radius = 0.2];
   \node[right] at (11,2) {point sources};
   \node[right] at (11,0) {location: $\omega_j$};
   \node[right] at (11,-2) {amplitude: $x_j(t)$};

   \draw[<->] (-6,-6) -- (10,-6);
   \node[above] at (2,-6) {far field};
   
   \draw[thick] (-5,-2) -- (-5,2);
   \node at (-4.7,1.5) {$\rhd$};
   \node at (-4.7,0.5) {$\rhd$};
   \node at (-4.7,-0.5) {$\rhd$};
   \node at (-4.7,-1.5) {$\rhd$};
   
   \node[left] at (-6,0) {sensors};
\end{tikzpicture}
}
\subfigure[inverse scattering]{
\includegraphics[width = 0.4\textwidth]{FigDemo/InverseScattering2.png}
}
\subfigure[Direction-Of-Arrival (DOA)]{
\includegraphics[width = 0.4\textwidth]{FigDemo/DOA.png}
}
\caption{Source localization in array imaging}
\label{FigSource}
\end{figure}

\begin{enumerate}

\item Inverse source and inverse scattering \cite{FannjiangYan,FannjiangScatteringI,FannjiangScatteringII}: in remote sensing, we assume there are $s$ sources located at $\{\om_j\}_{j=1}^s$ with amplitudes $\{x_j(t)\}_{j=1}^s$, and set a uniform array of sensors to collect measurements $\{y_e(t), \ t \in \Gamma\}$. Our goal is to recover the source locations and the sensor calibration parameters simultaneously from the measurements.


\item Direction-Of-Arrival (DOA) estimation \cite{ArrayImaging,SchD,Sch}: in signal processing literature, direction of arrival (DOA) denotes the direction from which usually a propagating wave arrives at a point, where usually an array of sensors (also called antenna elements) are located (see Figure \ref{FigSource} (b)). The goal is to find the direction relative to the array where the wave propagates from, which has
wide-range applications such as radar, sonar and computed tomography. By letting $\omega_j =  \Delta/\lambda \cos \phi_j$ where $\Delta$ is the spacing between sensors, $\lambda$ is the wave length, and $\phi_j$ is the direction of the $j$-th wave, we can model the DOA estimation exactly by the inverse problem in \eqref{eq:model3}.

\item Spectral analysis \cite{stoica1997introduction}: spectral analysis considers the problem of determining the spectral content (i.e., the distribution of power over frequency) of a time series from a finite set of measurements. It finds applications in many diverse fields, such as economics, astronomy, medicine, seismology, etc. For example, in speech analysis, spectral models of voice signals are useful in better
understanding the speech production process, and - in addition - can be used for both speech synthesis (or compression) and speech recognition. When the signals dealt with can be well described by a superposition of sines and cosines, classical spectral analysis is exactly the same as the inverse problem in \eqref{eq:model3}.
\end{enumerate} 
}

When all sensors are perfectly calibrated ($g$ is known), many methods have been proposed to recover frequencies located on a continuous domain, such as Prony's method \cite{Prony}, MUSIC \cite{SchD,Sch}, ESPRIT \cite{Esprit}, $\ell_1$ minimization \cite{CandesCPAM,Tang}, and  greedy algorithms \cite{SpectralCS,FL}. We refer the reader to \cite{Eldar_SamplingBook,Stoica} for a comprehensive review.
In this paper, the problem is complicated by the fact that each sensor has an unknown calibration parameter.

The calibration problem modeled by \eqref{eqprob1} 
has been considered in \cite{PKailath,WRS93, WRM94, FriedlanderWeiss} with the assumptions that the underlying frequencies have random and uncorrelated amplitudes, frequencies and noises are independent, and noises at different sensors are uncorrelated. 
In  \cite{PKailath}, Paulraj and Kailath investigated DOA estimation using a uniform linear array in the presence of unknown calibration parameters. By exploiting the Toeplitz structure of the measurement covariance in the case of perfect sensors (which we refer to as the case in which all calibration parameters are equal to $1$), they proposed an algebraic method relying on a least squares solution of two linear systems of equations for the calibration amplitudes and phases, respectively. 
However, the issue of phase wrapping in the set of equations for the calibration phase estimation is not taken into account and can degrade the DOA estimator performance (see Section \ref{SecAlgebraicFull}). This method is called the full algebraic method in our paper.
A similar approach is followed in \cite{WRM94}. 
In \cite{LiEr} Li and Er showed that, the bias in the full algebraic method with finite snapshots of measurements is nonzero, and if the problem of phase wrapping is resolved, then the variance for the calibration phases is $O(1/\sqrt L)$ where $L$ is the number of snapshots.
However, it remains unclear how to resolve phase wrapping in the full algebraic method.
Furthermore, sensitivity to noise and sensitivity in spectral estimation are not treated in \cite{LiEr}.
In \cite{WRS93}, Wylie, Roy, and Schmitt proposed a partial algebraic technique
which successfully avoids the problem of phase wrapping by removing redundancy in the linear system. A shortcoming of their approach is that a large part of the measurements are not used in the recovery process. 

In \cite{FriedlanderWeiss}, an alternating algorithm is proposed by Friedlander and Weiss for the same calibration problem modeled by \eqref{eqprob1}.
This algorithm is based on a two-step procedure. First, one assumes that the calibration parameters are known, and estimates the frequencies through the MUSIC algorithm. Then one solves an optimization problem to obtain the best calibration parameters based on the recovered frequencies. However, no performance guarantee for this approach is provided. In addition, this algorithm does not perform as well as the partial algebraic method in the presence of noise.

%



\subsection{Our contributions}
We begin by studying uniqueness of the calibration problem given by \eqref{eqprob1} and show that there are certain inevitable ambiguities in this problem. We characterize all trivial ambiguities, and prove that, both the spectrum and the calibration parameters are uniquely determined from infinite snapshots of noiseless measurements up to a trivial ambiguity, as long as there are more sensors than frequencies. Our proof is based on the algebraic methods proposed in \cite{PKailath,WRS93}. 
We then present a sensitivity analysis of the partial algebraic method \cite{WRS93} with respect to the number of snapshots $L$ and noise level $\sigma$. In particular, Theorem \ref{thmagstability} shows that, if the underlying frequencies are separated by $1/N$ (the standard resolution in spectral estimation), then the reconstruction error of the calibration parameters in the partial algebraic method is $O\left(L^{-\frac 1 2}(C_1+C_2\max(\sigma,\sigma^2))\right)$ for some constants $C_1,C_2>0$. This rate is verified by numerical experiments.
As for frequency localization, we prove that, the noise-space correlation function whose $s$ smallest local minima correspond to the recovered frequencies in the MUSIC algorithm, is perturbed by at most $O\left( L^{-\frac 1 2}(C_3+C_4\max(\sigma,\sigma^2))\right)$ for some constants $C_3,C_4>0$.
%


The partial algebraic method in \cite{WRS93} only exploits partial entries of the measurement covariance matrix and while full algebraic method in \cite{PKailath} is affected by phase wrapping. We therefore propose an optimization approach to make full use of the measurements.
We consider an objective function composed of two terms: one is a quadratic loss and the other is a penalty which prevents calibration parameters going to $\infty$ and frequency amplitudes decreasing to $0$ (or vice versa). This objective is continuously differentiable but non-convex. 
We propose to minimize it over all possible calibration parameters and Toeplitz matrices by Wirtinger gradient descent \cite{CandesWirtinger}. We prove that, Wirtinger gradient descent converges to a critical point, and show empirically that this critical point provides a good approximation to the true calibration parameters and the underlying frequencies.

Finally we perform a systematic numerical study to compare the partial algebraic method \cite{WRS93}, the alternating algorithm \cite{FriedlanderWeiss}, and our optimization approach. With respect to the stability to $L$ and $\sigma$, our algorithm has the best performance.

In summary, the main contributions of this paper are (i) characterizing all trivial ambiguities of the calibration problem modeled by \eqref{eqprob1} and proving uniqueness  up to a trivial ambiguity; (ii) presenting a sensitivity analysis of the partial algebraic method; (iii) proposing an optimization approach with superior numerical performance over previous methods.

\subsection{Related work}


Recently, many works addressed the following single-snapshot calibration problem: 
\beq
\label{modelling}
y = \diag(g) \Phi x_0 + e
\eeq
where $g \in \CC^{N}, x_0 \in \CC^M$ are the unknown calibration parameters and signal of interest respectively, $\Phi \in \CC^{N \times M}$ is a given sensing matrix, $e \in \CC^N$ represents noise, and $y \in \CC^{N}$ is the measurement vector \cite{RombergBD,LingStrohmer,LLSW}. The goal is to recover $g$ and $x_0$ from $y$.  Without additional assumptions, solutions to this problem are not unique since there are more unknowns than equations.  In \cite{RombergBD}, Ahmed, Recht, and Romberg assumed that $g$ lies in a known subspace, and used a lifting technique to transform the problem into that of recovering a rank-one matrix from an underdetermined system of linear equations.
They proved explicit conditions under which nuclear norm minimization is guaranteed to recover the original solution in the case where $\Phi$ is a random Gaussian matrix. 
In \cite{LingStrohmer}, Ling and Strohmer extended the framework in \cite{RombergBD} to allow for sparse signals, and a random Gaussian matrix or a random partial Discrete Fourier Transform (DFT) matrix.
The latter is closely related to spectral estimation assuming a discretized frequency grid with spacing equal to $1/N$. It is well known that when the underlying frequencies are on a continuous domain, this discretization may cause a large gridding error 
\cite{ChiMismatch,SpectralCS,FLSPIE,FL}. Here we do not discretize the frequencies in order to avoid gridding errors.
The lifting technique has been applied on a wide range of blind deconvolution and sensor calibration models \cite{Remi,ADemanet,cosse2017blind,ChiLift,yang2016super}, among which \cite{ChiLift,yang2016super} are mostly related with our model.
In \cite{ChiLift}, Chi considered a slightly different single-snapshot model: 
\beq
\label{modelchi}
y = \diag(g) A x
\eeq 
where $g \in \CC^{N}$ contains unknown calibration parameters, $A$ is the same as in \eqref{eqA}, $x \in \CC^s$ is an unknown amplitude vector, and $y \in \CC^{N}$ is the measurement vector. The goal is to recover $g$, $x$, and the frequencies $\{\om_j\}_{j=1}^s$ from $y$. 
This problem is the same as our model in \eqref{eqprob1} with a single snapshot. Chi solved \eqref{modelchi} using a lifting technique and atomic norm minimization, and proved that, in the noiseless case where $g$ lies in a  random subspace of dimension $K$ with coherence parameter $\mu$, and if the underlying frequencies are separated by $4/N$, then exact recovery is guaranteed with high probability as long as $N \ge C \mu K^2 s^2$ up to a log factor. Chi's result is generalized by Yang et al. in \cite{yang2016super} to the model
\beq
y(n)=\sum_{j=1}^s x_je^{2\pi n i w_j}g_j(n), \ n=0,\cdots, N-1,
\eeq  
where $x_j \in \mathbb{C}$, $w_j \in [0,1)$ and $\{g_j(n)\}$ are the unknown amplitude, location, and samples of the waveform associated with the $j$th complex exponential. It is assumed that (i) all $g_j$ live in a common random subspace of dimension $K$ with coherence parameter $\mu$; (ii) the $\omega_j$'s are separated by $4/N$. Yang et al. proved exact recovery with high probability as long as $N \ge C \mu sK$ up to a log factor. Noise was not considered in \cite{ChiLift,yang2016super}.

Since the lifting technique greatly increases computational complexity, many non-convex optimization approaches have been proposed to address problems in signal processing, such as phase retrieval \cite{CandesWirtinger,BE_STFT,WangTruncatedAmplitudeFlow,SunPR}, dictionary recovery \cite{SunDictionary}, blind deconvolution \cite{LLSW}, and low-rank matrix estimation \cite{TuoLowRank}. 
In \cite{LLSW}, Li et al. formulated a non-convex optimization problem for the calibration problem modeled by \eqref{modelling}, and solved it with a two-step scheme composed of a good initialization and gradient descent \cite{LLSW}. Performance guarantees were proved when $g$ lies in a known subspace and $\Phi$ is a random Gaussian matrix. This theory does not apply to spectral estimation since our sensing matrix is not random Gaussian.

After sensors are built, it is usually cheap to take multiple snapshots of measurements. This paper studies the calibration problem modeled by \eqref{eqprob1} with multiple snapshots. In comparison with the works in \cite{RombergBD,ChiLift,LingStrohmer,LLSW}, we remove the assumption that $g$ lies in a known subspace. Instead we utilize the Fourier structure and take advantage of the multiple snapshots. In addition, we only require more sensors than frequencies, namely, $N >s$.




%
 To consider multiple snapshots of measurements, the works in \cite{LiBresler,LiBreslerPCA} addressed the following calibration model: 
\beq
\label{modelli}
Y = \diag(g) \Phi X_0 + E
\eeq
where $g \in \CC^{N},X_0 \in \CC^{M \times L}$ are the unknown calibration parameters and signals in $L$ snapshots respectively, $\Phi \in \CC^{N \times M}$ is a given sensing matrix, $E \in \CC^{N \times L}$ represents noise, and $Y \in \CC^{N \times L}$ is the measurement matrix in $L$ snapshots.
In \cite{LiBresler}, Li, Lee and Bresler proved uniqueness up to a scaling ambiguity for generic signals $X_0$ and a generic sensing matrix $\Phi$ provided that $N > M$ and $\frac{N-1}{N-M} \le L \le M$. In the case where $X_0$ has $s$ non-zero rows, uniqueness was proved for generic signals with $s$ non-zero rows and a generic sensing matrix $\Phi$ provided that $N > 2s$ and $\frac{N-1}{N-2s} \le L \le s$. These conditions are optimal in terms of sample complexity. 
In \cite{LiBreslerPCA}, the authors formulated this calibration problem as an eigenvalue/eigenvector problem, and solved it via power iterations, or when $X$ is sparse or jointly sparse, via truncated power iterations. In \cite{lee2018fast}, the same kind of power method was applied to solve a  multichannel blind deconvolution problem. 
%
However, the theory and algorithms in \cite{lee2018fast, LiBresler,LiBreslerPCA} do not apply to our case for the following reasons: (i) our sensing matrix defined in \eqref{eqA} is unknown, since it depends on the unknown frequencies; (ii) even if we discretize the frequency domain and approximate every frequency by the nearest grid point to fit the model \eqref{modelli}, our sensing matrix is not generic. In comparison, our uniqueness results consider the Fourier structure of the measurements but assume infinite snapshots. It is also interesting to study the optimal sample complexity in our case; we leave this problem for future research. 
\commentout{
Later in \cite{LiBresler1}, Li, Lee and Bresler linearize this problem by reformulating the equation \eqref{modelli} into 
\beq
\label{modelli1}
\diag(\gamma)Y-A X_0 =0
\eeq
where $\gamma$ is the point-wise inverse of the calibration parameters $g$.  They  use a lifting technique to transform \eqref{modelli1} to an eigenvalue/eigenvector problem and  then propose to solve it  via power iterations when $X_0$ has  a subspace structure or via truncated power iteration when $X_0$ has a joint sparsity structure. 
}


\subsection{Organization and Notation}
This paper is organized as follows.  
Uniqueness results are described in Section \ref{sec:uniqueness}. The partial algebraic method and its sensitivity analysis are presented in Section \ref{sec:algebraic}. Section \ref{sec:optimization} considers our optimization approach and its convergence to a critical point. Numerical simulations are presented in Section \ref{sec:numerics}. We conclude and discuss future research directions in Section \ref{sec:conclusion}. Most of the proofs are relegated to the Appendices.


Throughout the paper we use $s,N,L$ to denote the number of frequencies, sensors, and snapshots respectively.
The expression $C = (A\ B)$ horizontally concatenates matrices A and B, while $C = (A; B)$ concatenates them vertically.
For $x \in \CC^N$, $\sum x :=  \sum_{j=1}^N x_j$, and $\diag(x)$ is the $N \times N$ diagonal matrix whose diagonal is $x$. We use $|x| \in \RR^N$ and $\angle x \in \RR^N$ to denote the magnitude and phase vectors of $x$ respectively such that $|x|_j = |x_j|$ and $(\angle x)_j = \angle x_j, j=1,\ldots,N$. 
The dynamic range of $x$ is the ratio between the largest amplitude and the smallest amplitude of $x$, denoted by $\DR_x:=\max_i |x_i| / \min_i |x_i|$.  
For $X \in \CC^{N \times N}$, we use $\diag(X)$ to denote the main diagonal of $X$, $\diag(X,k), k >0$ to denote the $k$th diagonal of $X$ above the main diagonal, and $\diag(X,k), k<0$ to denote the $k$th diagonal of $X$ below the main diagonal. The notation $\|X\|$ and $\|X\|_F$ represent the spectral norm and the Frobenius norm of $X$, respectively.
The expression $A \preceq B$ for two square matrices $A$ and $B$ of the same size means $B-A$ is positive semidefinite. The expression $a \lesssim b$ for two scalars $a$ and $b$ means $a \le C b$ with a constant $C$ independent of $a,b$.
We use $\mathbf{0}$ to denote the null vector or the null matrix. 


\section{Uniqueness results}

\subsection{Trivial ambiguity and uniqueness}

In the case of a single snapshot, uniqueness does not exist since there are fewer measurements $(N)$ than unknowns $(N+2s)$.
After sensors are built, it is often cheap to take multiple snapshots.
Therefore, for uniqueness we study the setting where infinite snapshots of noiseless measurements are taken. 
Clearly certain trivial ambiguities between the spectrum and  the calibration parameters are inevitable. For example, one can add a gain to $x(t)$ and then divide it out in $g$, so that there is always a gain ambiguity. Similarly, there is always a shift ambiguity in the frequencies since we can shift all frequencies by a constant and all this will do is add a phase modulation to $g$. We define trivial ambiguities in the calibration problem as follows:
\begin{definition}[Trivial ambiguity]
Let $\{g, \calS,x(t)\}$ be a solution to the calibration problem modeled by \eqref{eqprob1}. Then $\{\tilde g, \tilde\calS,\tilde{x}(t) \}$ is called equivalent to $\{g, \calS,x(t)\}$ up to a trivial ambiguity if there exist $c_0  >0, c_1,c_2 \in \RR$ such that 
\begin{align*}
\tilde g & =\{\tilde g_n = c_0 e^{i (c_1+n c_2)}g_n\}_{n=0}^{N-1}
\\
\tilde\calS &= \{\tilde\om_j: \tilde\om_j = \om_j - c_2/(2\pi)\}_{j=1}^s
\\
\tilde{x}(t) &= x(t) c_0^{-1}e^{-i c_1}.
\end{align*} 
\end{definition}

To proceed we make the following assumptions on the statistics of the frequencies (or sources in array imaging) and noises:
\begin{description}

\item[A1] Calibration parameters do not vanish: $|g_n| \neq 0, \ n=0,\ldots,N-1$.

\item[A2] Sources and noises have zero mean: $\EE x(t) = \mathbf{0}$ and $\EE e(t) = \mathbf{0}$.

\item[A3] 
Sources are uncorrelated such that 
$
R^x: = \EE x(t)x^* (t) = \diag(\{\gamma_j^2\}_{j=1}^s).
$

\item[A4] Sources and noises are independent, i.e.,  $\EE x(t)e^*(t) = \mathbf{0}$.

\item[A5] Noises at different sensors are uncorrelated so that $\EE e(t)e^*(t) = \sigma^2 I_N$ where $\sigma$ represents the noise level.

\end{description}

\label{sec:uniqueness}

Define 
\beq
\label{eqf}
f_n =\frac{1}{N} \sum_{j=1}^s \gamma_j^2 {e^{2\pi i n\om_j}}, \ n=0,\ldots,N-1.
\eeq
The values $\{f_n\}_{n=0}^{N-1}$ contain sufficient information to recover all frequencies by standard spectral estimation. Notice that $f_0 = \sum_{j=1}^s \gamma_j^2>0$. 
In the case of infinite snapshots, uniqueness of the calibration problem exists up to a trivial ambiguity as long as $|f_1|>0$. This is a sufficient condition to guarantee that the sub-diagonal entries in the covariance measurement do not vanish. If the amplitudes $\{\gamma_j^2\}_{j=1}^s$ are generic, then we have $|f_1|>0$ almost surely.
\begin{theorem}
\label{thm1}
Suppose $|f_1| > 0$, $N \ge s+1$, and the assumptions A1-A5 hold. Let $\{g, \calS,x(t)\}$ be a solution to the calibration problem modeled by \eqref{eqprob1}.
If there is another solution $\{\tilde g,\tilde\calS,\tilde{x}(t)\}$, then $\{\tilde g,\tilde\calS,\tilde{x}(t)\}$ is equivalent to $\{g, \calS,x(t)\}$ up to a trivial ambiguity.
\end{theorem}

\subsection{Proof of uniqueness}

We prove Theorem \ref{thm1} based on the algebraic technique proposed in \cite{PKailath,WRS93}. 
We begin by considering the covariance of $y_e$:
\beq
\label{eqcovRye}
R^y_e := \EE y_e(t) y_e^*(t)
= G A \underbrace{\EE x(t)x^*(t)}_{R^x} A^* G^* 
+ G A \underbrace{\EE x(t) e^*(t)}_{R^{xe}} 
+\underbrace{ \EE e(t)x^*(t)}_{R^{ex}} A^* G^* 
+ \underbrace{\EE e(t)e^* (t)}_{R^e}.
\eeq
Denote the noiseless data by $y(t) := G A x(t)$ and its covariance by 
\beq
\label{eqRy}
R^y := \EE y(t)y^*(t) = GA R^x A^* G^*.
\eeq
Under Assumptions (A1-A5), we have
\beq
\label{eqcovRye2}
R^y_e = R^y + R^e = GA R^x A^* G^* + \sigma^2 I_N.
\eeq
In the noiseless case, $\sigma =0$ and $R^y_e = R^y$. If infinite snapshots are collected, then we can assume $R^y$ is known.

Define the Toeplitz operator which maps a sequence to a Toeplitz matrix: $$\calT: \RR \times \CC^{N-1} \rightarrow \CC^{N \times N}: \ \calT(f) :=\begin{bmatrix}
f_0 & \barf_1 & \ddots & \barf_{N-2} & \barf_{N-1} \\ 
f_1& f_0 & \ddots & \ddots & \barf_{N-2} \\ 
\ddots & \ddots & f_0 & \ddots & \ddots \\ 
f_{N-2} & \ddots & \ddots & \ddots & \barf_{1} \\ 
f_{N-1} & f_{N-2} & \ddots & f_1 & f_{0} \\ 
\end{bmatrix}.
$$
With this notation
\beq
\label{eqRy2}
R^y = \diag(g) \calT(f) \diag(\barg),
\eeq
where $f$ is the sequence defined in \eqref{eqf}. The $(m,n)$th entry of $R^y$ is 
\beq
\label{eqRymn}
R^y_{m,n} =\frac{g_m \bar{g}_n}{N} \sum_{j=1}^s \gamma_j^2  {e^{2\pi i \om_j (m-n)}}= g_m \barg_n f_{m-n}.
\eeq
Throughout the paper we write $g_n = \alpha_n e^{i\beta_n}$, where $\alpha_n$ is the calibration amplitude and $\beta_n$ is the calibration phase at the $n$th sensor. By \eqref{eqRymn}, all calibration amplitudes can be uniquely determined from the diagonal entries of $R^y$ up to a scaling, and the calibration phases can be determined from the subdiagonal of $R^y$ up to a trivial ambiguity as long as $f_1$ does not vanish.


\begin{lemma}
\label{lemma2}
Suppose $|f_1| > 0$ and the assumptions A1-A5 hold. If there is another set $\{\tilde g \in \CC^N,\tilde f \in \RR \times \CC^{N-1}\}$ satisfying $\diag(\tilde g)\calT(\tilde f)\diag(\bar{\tilde g}) = \diag(g)\calT(f)\diag(\barg)$, then there exist $c_0>0$ and $c_1,c_2\in \RR$ such that
\begin{align*}
\tilde g_n  &= c_0 e^{i (c_1+n c_2)} g_n, \quad \tilde f_n = c_0^{-2} e^{-in  c_2} f_n, ~~ n=0,\ldots,N-1.
\end{align*}
\end{lemma}

\begin{proof}
We write $g_n = \alpha_n e^{i\beta_n}$ where $\alpha_n$ is the calibration amplitude and $\beta_n$ is the calibration phase at the $n$th sensor. Then $\alpha = |g|$ and $\beta = \angle g$. Similarly, let  $\talpha = |\tg|$ and $\tbeta = \angle \tg$. 
Observe that the diagonal entries of $R^y$ are
\begin{align*}
R^y_{n,n}  = \alpha_n^2  f_0, \ n=0,\ldots,N-1.
\end{align*}
If $f_0$ is given, then $\talpha_n = \sqrt{R^y_{n,n}/f_0}$; otherwise, the unknown $f_0$ leads to a scaling ambiguity such that $\talpha = c_0 \alpha$ for some constant $c_0>0$.

The sub-diagonal entries of $R^y$ are
$$
R^y_{n,n-1} = \alpha_n \alpha_{n-1} e^{i(\beta_n - \beta_{n-1})} f_1 \neq 0, \ n=1,\ldots,N-1.
$$
This gives rise to $N-2$ equations regarding the $\beta_n$'s:
$$
e^{i(\beta_{n+1} - 2 \beta_n + \beta_{n-1})}  =  \frac{\alpha_{n-1} R^y_{n+1,n}}{\alpha_{n+1} R^y_{n,n-1}},
$$
which are equivalent to
\beq
\beta_{n+1} - 2 \beta_n + \beta_{n-1}
 = \angle \frac{R^y_{n+1,n}}{R^y_{n,n-1}}+2\pi k_n, \ n = 1,\ldots,N-2, \ k_n \in \ZZ.
\label{linearsystembeta1}
\eeq
\commentout{
Equivalently, the $\beta_n$'s satisfy the following linear system:
\beq
\label{linearsystembeta1}
\begin{bmatrix}
1 & -2 & 1 & 0 & 0 &  \ldots  & 0& 0  & 0 
\\
0 & 1 & -2 & 1 & 0  &  \ldots  & 0 & 0  & 0 
\\
0 & 0 & 1 & -2 & 1   &  \ldots  & 0 & 0  & 0
\\
 &  &  &  & \ldots   &    &  &   & 
 \\
  &  &  &  & \ldots   &    &  &   & 
  \\
   &  &  &  & \ldots   &    &  &   & 
 \\ 
 0 & 0 & 0 & 0 & 0 &   \ldots  & 1 & -2  & 1
\end{bmatrix}
\begin{bmatrix}
\beta_0 
\\
\beta_1
\\
\beta_2
\\
\beta_3
\\
\beta_4
\\
\vdots
\\
\beta_{N-3}
\\
\beta_{N-2}
\\
\beta_{N-1}
\end{bmatrix}
-
\begin{bmatrix}
\angle (R^y_{2,1}/R^y_{1,0})
\\
\angle (R^y_{3,2}/R^y_{2,1})
\\
\angle (R^y_{4,3}/R^y_{3,2})
\\
\vdots
\\
\vdots
\\
\vdots
\\
\angle (R^y_{N-1,N-2}/R^y_{N-2,N-1})
\end{bmatrix} 
=\mathbf{0} \mod  2\pi.
\eeq
}
The linear system for $\beta$ given by \eqref{linearsystembeta1} has $N-2$ independent equations and $N$ variables. Solving  \eqref{linearsystembeta1} results in 
\begin{align*}
\tbeta=\beta+
 c_1 
\begin{bmatrix}
1 
\\
1
\\
\vdots
\\
1
\\
1
\end{bmatrix}
+  c_2
\begin{bmatrix}
0 
\\
1
\\
\vdots
\\
N-2
\\
N-1
\end{bmatrix}
 \mod 2\pi.
\end{align*}
%
Therefore, $\tbeta_n = \beta_n + c_1 + nc_2 \mod 2\pi$. Combined with $\talpha = c_0 \alpha$, we have $\tg_n = c_0 e^{i( c_1+n c_2)}  g_n$. As for $\tf$, since
$R^y_{m,n} = g_m \barg_n f_{m-n} = \tg_m \bar{\tg}_n \tf_{m-n},$
we obtain 
$\tf_{m-n} = c_0^{-2} e^{-i(m-n) c_2}f_{m-n},$
which concludes the proof.
\end{proof}

After obtaining the calibration parameters $\{\tilde g_n = c_0 e^{i(c_1+n c_2)}g_n\}_{n=0}^{N-1}$, we simply divide $\tg$ out of $R^y$ to obtain 
$$\tilde F = \diag(\tg)^{-1}R^y \diag(\bar{\tg})^{-1} = c_0^{-2} D_{c_2} A R^x A^* D_{c_2}^*= c_0^{-2} \tilde A R^x {\tilde A}^*,$$
where $D_{c_2} =  \diag\left(\{e^{-inc_2}\}_{n=0}^{N-1} \right)$ and
\beq
\label{eqtA}
\tilde A_{n,j} \in \CC^{N \times s}:\ \tilde A_{n,j} = \frac{1}{\sqrt{N}}e^{2\pi i n (\om_j - \frac{c_2}{2\pi})}. 
\eeq
We can then perform standard spectral estimation on $\tilde F$ using the MUltiple Signal Classification (MUSIC) algorithm proposed by Schmidt \cite{SchD,Sch} to retrieve the frequencies. 
MUSIC is introduced in Section \ref{sec:algebraic}. For now, we use the result that MUSIC guarantees exact recovery of frequencies in the noiseless case as long as $N \ge s+1$ (see Proposition \ref{propmusic1}). Combining this with Lemma \ref{lemma2} gives rise to Theorem \ref{thm1}.

\subsection{A general condition to guarantee uniqueness}
\label{secuniquegeneral}
Theorem \ref{thm1} guarantees uniqueness when $ f_1 \neq 0$. This condition can be generalized as follows. For $k=1,\ldots N-1$, we have 
$R^y_{l+k,l} = \alpha_{l+k}\alpha_l e^{i(\beta_{l+k}-\beta_l)}f_k$ where $l = 0,\ldots,N-k-1$.
As long as $|f_k| \neq 0$, we can compute $R^y_{l+k+1,l+1}/R^y_{l+k,l}$ for $l=0,\ldots,N-k-2$ and obtain the following system of linear equations 
\begin{align}
\beta_{l+k+1}-\beta_{l+k}-\beta_{l+1}+\beta_l \equiv \angle R^y_{l+k+1,l+1}/R^y_{l+k,l}  \mod 2\pi
\label{linearsystembeta123}
\end{align}
for $l=0,\ldots,N-k-2$. Here $a\equiv b \mod c$ means $a$ is equal to $b$ modulo $c$. 
We may write these $N-k-1$ linear equations \eqref{linearsystembeta123} in matrix form as
\[\Phi_k \beta \equiv  \delta_{k} \mod 2\pi, \]
where
$$
\Phi_1 = 
\begin{bmatrix}
1 & -2 & 1 & 0 &  \ldots  & 0& 0  & 0 
\\
0 & 1 & -2 & 1 &   \ldots  & 0 & 0  & 0 
\\
 &  &  &   \ldots   &    &  &   & 
 \\
  &  &  &   \ldots   &    &  &   & 
  \\
   &  &  &   \ldots   &    &  &   & 
 \\ 
 0 & 0 & 0 & 0 &    \ldots  & 1 & -2  & 1
 \end{bmatrix} \in \CC^{(N-2) \times N},
$$
 $$\Phi_k=\begin{bmatrix}
 1 &-1&\onegroup  & -1& 1&  & &  \\ 
 &1& -1&\onegroup & -1& 1& &  \\
 &  &\ddots& \ddots& \ddots& \ddots &\ddots \\
   && &1&-1&\onegroup&-1&1
\end{bmatrix} \in \CC^{ (N-k-1) \times N}$$ 
for $k=2,\ldots N-2$ and 
$$\delta_k = [\angle R^y_{k+1,1}/R^y_{k,0} \quad  \angle R^y_{k+2,2}/R^y_{1+k,1} \ \ldots \  \angle R^y_{N-1,N-k-1}/R^y_{N-2,N-k-1}]^T 
\in \CC^{N-k-1}.$$

Let $\Lambda=\{k: |f_k| \neq 0, k=1,\cdots, N-2\}$. Concatenating  the matrices  $\Phi_k$   for $k\in \Lambda$ vertically yields the following matrix 
$$
\Phi_{\Lambda}=\begin{bmatrix}
\vdots\\
 \Phi_{k} \\ 
 \vdots\\
\end{bmatrix}_{k \in \Lambda}.$$
By exploiting all entries in the covariance matrix, we can generalize the condition $|f_1|>0$ in Theorem \ref{thm1} to the condition that rank$(\Phi_{\Lambda})=N-2$. Thus, uniqueness in Theorem \ref{thm1} holds under the more general condition: rank$(\Phi_{\Lambda})=N-2$. Observe that $|f_1|>0$ is sufficient but not necessary to guarantee rank$(\Phi_{\Lambda})=N-2$.

\section{Algebraic methods}
\label{sec:algebraic}

\subsection{Full algebraic method and phase wrapping}
\label{SecAlgebraicFull}

In \cite{PKailath}, Paulraj and Kailath proposed the first method for sensor calibration in DOA estimation. By exploiting the Toeplitz structure of $\calT(f)$, they obtained two linear systems for calibration amplitudes and phases, respectively.

Consider the case with infinite snapshots of noiseless measurements.
When $k$ varies from $1$ to $N-1$, the $k$-th sub-diagonal entries of $R^y$ satisfy
$$
\alpha_{l+k}\alpha_l e^{i (\beta_{l+k}-\beta_l)}f_k= R^y_{l+k,l}, \ l=0,\ldots,N-k-1.
$$
For $k=1,\ldots, N-1$, as long as $|f_k| \neq 0$, one can obtain the following system of equations for $\ln \alpha$
\begin{align}
&\ln \alpha_{l+k+1} + \ln \alpha_{l+1}
- \ln \alpha_{l+k} - \ln \alpha_l 
= \ln \frac{|R^y_{l+k+1,l+1}|}{|R^y_{l+k,l}|}
\label{linearsystemlnalpha}
\end{align}
where $l=0,\ldots,N-k-2,$
as well as \eqref{linearsystembeta123} for calibration phases $\beta$.
Paulraj and Kailath proposed to substitute $\equiv$ (equal modulo $2\pi$) with $=$ in \eqref{linearsystembeta123}
and
solve these linear systems by least squares.
However, one has to consider phase wrapping in \eqref{linearsystembeta123} so that \eqref{linearsystembeta123} is equivalent to
\begin{align}
&\beta_{l+k+1}-\beta_{l+k}-\beta_{l+1}+\beta_l = \angle R^y_{l+k+1,l+1}/R^y_{l+k,l} + 2\pi p_{k,l} \text{  where  } \ p_{k,l} \in \ZZ
\label{linearsystembeta13}
\end{align}
where $l=0,\ldots,N-k-2 \text{ and } k=1,\ldots N-1$. Importantly, the $p_{k,l}$'s in \eqref{linearsystembeta13} are not independent. For example, the parameters $p_{1,0},p_{1,1}, p_{2,0}$ need to satisfy
\beq
\angle R^y_{2,1}/R^y_{1,0}+2\pi p_{1,0}+ \angle R^y_{3,2}/R^y_{2,1}+2\pi p_{1,1} = \angle R^y_{3,1}/R^y_{2,0} + 2\pi p_{2,0}.
\label{prelation}
\eeq
The $p_{k,l}$'s are constrained by many more equations like \eqref{prelation}.
Solving \eqref{linearsystembeta13} with the constraints involves phase synchronization, which itself is highly nontrivial.

\subsection{Partial algebraic method}

\renewcommand{\algorithmicrequire}{\textbf{Input:}}    \renewcommand{\algorithmicensure}{\textbf{Output:}}
\setlist[description]{font=\normalfont} 
\begin{algorithm}[t!]                      	
\caption{Partial algebraic method}          	
\label{algorithm1}		
\begin{algorithmic}[1]                    	
    \REQUIRE  Measurements $\{y_e(t), t\in \Gamma\}$ and sparsity $s$. 
    \ENSURE Calibration parameters $\widehat{g} = \{\widehat{g}_n := \halpha_n e^{i\hbeta_n} \}_{n=0}^{N-1}$ and recovered spectrum $\{\widehat{\om}_j\}_{j=1}^s$
    \STATE Form the empirical covariance matrix 
    \begin{equation}
    \label{eqempiricalcov}
    \widetilde R^y_e = \frac 1 L  \sum_{t \in \Gamma} y_e(t)y_e^*(t).
    \end{equation} 
    \STATE Compute the eigenvalue decomposition: $$\widetilde  R^y_e = U \Sigma U^*$$ where $\Sigma = \diag(\lam_0(\wtR^y_e),\ldots,\lam_{N-1}(\wtR^y_e)), \ \lam_0(\wtR^y_e) \ge \lam_1(\wtR^y_e) \ge \ldots.$
    \STATE Estimate the noise level $\hsigma = \sqrt{{\sum_{l=s}^{N-1}\lam_l(\wtR^y_e)}/(N-s)}$.
    \STATE Subtract the noise component: $\widehat R^y \leftarrow \widetilde R_e^y - \hsigma^2 I_N$.
    \STATE Estimate calibration amplitudes: $\halpha_n = \sqrt{\widehat R^y_{n,n}}, \ n=0,\ldots,N-1$.
    \STATE Solve the following linear system $\Phi \hbeta = \widehat b$ to obtain calibration phases $\hbeta$
    \beq
\label{linearsystembeta2}
\underbrace{
\begin{bmatrix}
\frac{1}{N^2} & 0 & 0 & 0  &  \ldots  & 0& 0  & 0 
\\
1 & -2 & 1 & 0 &  \ldots  & 0& 0  & 0 
\\
0 & 1 & -2 & 1 &   \ldots  & 0 & 0  & 0 
\\
 &  &  &   \ldots   &    &  &   & 
   \\
   &  &  &   \ldots   &    &  &   & 
 \\ 
 0 & 0 & 0 & 0 &    \ldots  & 1 & -2  & 1
 \\
  0 & 0 & 0 & 0 &    \ldots  & 0 & 0  & \frac{1}{N^2}
\end{bmatrix}
}_{\Phi \in \RR^{N \times N}}
\underbrace{
\begin{bmatrix}
\hbeta_0 
\\
\hbeta_1
\\
\hbeta_2
\\
\vdots
\\
\hbeta_{N-3}
\\
\hbeta_{N-2}
\\
\hbeta_{N-1}
\end{bmatrix}
}_{\widehat\beta \in \RR^{N}}
=
\underbrace{
\begin{bmatrix}
0
\\
\angle ({\widehat R}^y_{2,1}/{\widehat R}^y_{1,0})
\\
\angle ({\widehat R}^y_{3,2}/{\widehat R}^y_{2,1})
\\
\vdots
\\
\vdots
\\
\angle ({\widehat R}^y_{N-1,N-2}/{\widehat R}^y_{N-2,N-1})
\\
0
\end{bmatrix} 
}_{\widehat b \in \RR^{N}}.
\eeq

     \STATE Compute the matrix $\widehat{F} = \widehat{G}^{-1} \widehat R^y {\overline{\widehat{G}}}^{-1}$ where $\widehat{G} = \diag(\hg)$ and $\hg_n = \halpha_n e^{i \hbeta_n}, n=0,\ldots,N-1$.
     \STATE Apply the MUSIC algorithm on $\widehat{F}$:
    \begin{description}
     \item[i)] Compute the eigenvalue decomposition: $\widehat{F} = [V_1 \ V_2] \diag(\lam_1(\widehat F),\ldots,\lam_s(\widehat F),\ldots) [V_1 \ V_2]^*$ where $V_1 \in \CC^{N \times s}$, and $\lam_1(\widehat F) \ge \lam_2(\widehat F)\ge \ldots$.
      \item[ii)] Compute the imaging function $$\widehat{\calJ}(\om) = \frac{\|\phi(\om)\|}{\|V_2^*\phi(\om)\|}$$ where $\phi(\om) = [1 \ e^{2\pi i \om} \ \ldots e^{2\pi i(N-1) \om}]^T$.
      \item[iii)] Return the spectrum $\{\widehat\om_j\}_{j=1}^s$ corresponding to the $s$ largest local maxima of $\widehat{\calJ}(\om)$.
     \end{description}
\end{algorithmic}
\end{algorithm}

In \cite{WRS93}, Wylie, Roy and Schmitt proposed a partial algebraic method by 
using the system in \eqref{linearsystembeta1}, corresponding to the set of equations in \eqref{linearsystembeta123} with $k=1$, to recover  the calibration phases. These linear equations are independent so that there is no problem of phase wrapping. This partial algebraic method is  summarized in Algorithm \ref{algorithm1}.

In practice, we take $L$ snapshots of independent measurements, i.e., $\{y_e(t) , t \in \Gamma, \#\Gamma = L\}$, and approximate $R^y_e$ by the empirical covariance matrix $\wtR^y_e$ in \eqref{eqempiricalcov}.
The noise level $\sigma$ can be estimated from the smallest $N-s$ eigenvalues of $\wtR^y_e$ and the noise component can be subtracted from $\wtR^y_e$ to yield $\hR^y$ as an approximation of $R^y$ (see Step 4 in Algorithm \ref{algorithm1}). 
We then identify the calibration amplitudes $|g|$ from the diagonal entries of $\hR^y$ and calibration phases $\angle g$ from the sub-diagonal entries of $\hR^y$ by solving \eqref{linearsystembeta2} which gives a specific solution to \eqref{linearsystembeta1} with $\beta_0 = \beta_{N-1} = 0$. %
After all calibration parameters are recovered, MUSIC is applied for the usual spectral estimation, which guarantees exact recovery with exact data.

\begin{proposition}
\label{propmusic1}
Suppose $s$ is known and the input of MUSIC is exact: $F = A R^x A^*$. If $N \ge s+1$, then 
$$\om \in \calS 
\Longleftrightarrow
\calR(\om) = 0
\Longleftrightarrow \calJ(\om) = \infty$$ 
where $\calJ(\om)$ is the imaging function defined in Step 8.ii in Algorithm \ref{algorithm1} and  $\calR(\om) := 1/\calJ(\om)$ is called the noise-space correlation function.
\end{proposition}

In the noiseless case, one can extract the $s$ zeros of the noise-space correlation function $\calR(\om)$, or the $s$ largest local maxima of the imaging function $\calJ(\om)$ as an estimate of $\calS$.
In the presence of noise, suppose the input of MUSIC is approximate: $\widehat{F} = A R^x A^* + E$, and the noise-space correlation function is perturbed from $\calR$ to $\widehat\calR$. Stability of MUSIC depends on the perturbation of the
noise-space correlation function which can be estimated in the following lemma, thanks
to classical perturbation theory of singular subspaces by Wedin \cite[Theorem 3.4]{Wedin1972,StewartSun,WedinLi}.

\begin{proposition}
\label{propmusic2}
Let $N \ge s+1$. Suppose $s$ is known and the input of MUSIC is approximate: $\widehat{F} = F+E= A R^x A^* + E$. Let $\lam_1(F) \ge \ldots \ge \lam_s(F)$ be the nonzero eigenvalues of $F$.
Suppose $\calR(\om)$ and $\widehat{\calR}(\om)$ are the {noise-space correlation functions} when MUSIC is applied on $F$ and $\widehat{F}$ respectively.
If $2\|E\| < \lambda_s(F)$, then
$$\sup_{\om \in [0,1)}|\widehat\calR (\om) - \calR(\om)| \le \frac{2}{\lambda_s(F)}\|E\|.$$
\end{proposition}

\commentout{

\begin{algorithm}[ht]                      	
\caption{MUltiple SIgnal Classification (MUSIC)}          	
\label{algorithmMUSIC}		
\begin{algorithmic}[1]                    	
    \REQUIRE  $\widehat{F} \approx A R^x A^*$ and the sparsity $s$
    \ENSURE Spectrum $\{\widehat\om_j\}_{j=1}^s$ 
    \STATE Compute eigenvalue decomposition: $\widehat{F} = [V_1 \ V_2] \diag(\lam_1(\widehat F),\ldots,\lam_s(\widehat F),\ldots) [V_1 \ V_2]^*$ where $V_1 \in \CC^{N \times s}$, and $\lam_1(\widehat F) \ge \lam_2(\widehat F)\ge \ldots$.
    \STATE Compute imaging function $\calJ(\om) = \|\phi(\om)\|/\|V_2^*\phi(\om)\|$ with $\phi(\om) = [1 \ e^{2\pi i \om} \ \ldots e^{2\pi i(N-1) \om}]^T$
    \STATE Return the spectrum $\{\widehat\om_j\}_{j=1}^s$ corresponding to the $s$ largest local maxima of $\calJ(\om)$
    \end{algorithmic}
\end{algorithm}
}

\subsection{Sensitivity of the partial algebraic method}
\label{sec:algebraicsensitivity}

Theorem \ref{thm1} guarantees exact recovery up to a trivial ambiguity with infinite snapshots of noiseless measurements. In practice, only finite snapshots of noisy measurements are taken. We next present a sensitivity analysis of the partial algebraic method in Algorithm \ref{algorithm1} with respect to the number of snapshots $L$ and noise level $\sigma$. In particular, we prove that, there exist $C_1,C_2>0$ such that
$$\text{Reconstruction error of calibration parameters} \le O\left(\frac{C_1+C_2\max(\sigma,\sigma^2)}{\sqrt{L}}\right)$$
by the partial algebraic method when the true frequencies are separated by $1/N$. As for frequency localization in the MUSIC algorithm, the recovered frequencies correspond to the $s$ smallest local minima of the noise-space correlation function. Here we prove that, the noise-space correlation function is perturbed when using finite snapshots and in the presence of noise by at most $O\left( L^{-\frac 1 2}(C_3+C_4\max(\sigma,\sigma^2)) \right)$ for some constants $C_3,C_4>0$. The constants $C_1, C_2, C_3$ and $C_4$ depend on the number of sources $s$, the number of sensors $N$,  and dynamic ranges of the calibration parameters and source amplitudes. We will make these dependencies explicit in Remark \ref{remarkconstant}.
In the theorem below, let $\gammamax = \max_j \gamma_j, \gammamin = \min_j \gamma_j$, and $\alphamax = \max_n |g_n|, \alphamin = \min_n |g_n|$.  
\begin{theorem}
\label{thmagstability}
 In addition to the assumptions A1-A5, assume $N \ge s+1$, $|f_1| >0$ and the source and noise amplitudes $\|x(t)\|$ and $\|e(t)\|$ are almost surely bounded. 
 Let $\hR^y$ be the outcome in Step 4, $\hg$ be the recovered calibration parameters, and $\hF$ be the outcome in Step 7 of the partial algebraic method in Algorithm \ref{algorithm1}. 
Define
\begin{align}
\Delta R^y 
&:=
2 \alphamax^2\sigma^2_{\max}(A)
\left( \frac{\gammamax \max_{t \in \Gamma } \|x(t)\| \sqrt{2\log4 s}}{\sqrt{L}} + \frac{\gammamax^2 + \max_{t \in \Gamma} \|x(t)\|^2}{3L} \log 4s \right)
\nn
 \\
&+ 4\alphamax\smax(A)
\left( \frac{\sigma\gammamax\sqrt{2N\log (N+s)}}{\sqrt{L}} 
+ \frac{\max_{t \in \Gamma}\|x(t)\|\|e(t)\|}{3L} \log (N+s)
\right) \nn
\\
&
 +2 \left( \frac{\sigma \max_{t \in \Gamma} \|e(t)\|\sqrt{2\log 2N}}{\sqrt{L}}
+ \frac{\sigma^2+\max_{t\in \Gamma} \|e(t)\|^2}{3L}\log 2N \right).
\label{eqDeltaRy}
\end{align} 
Then $$\EE \|R^y - \hR^y\| \le \Delta R^y.$$
Let $\calR(\om)$ and $\widehat\calR(\om)$ be the noise-space correlation functions in MUSIC with the input data $F = AR^x A^*$ and $\widehat F$ respectively. Then
\begin{align}
\EE\min_{c_0>0,c_1,c_2\in \RR}
\max_n \|c_0 \hg_n -e^{i(c_1+nc_2)} g_n\|_\infty
& \le 
\left(
 \frac{3(\|g\|^2+N \alphamax^2)}{2\alphamin\|g\|^2 f_0 } 
+
144 N^2 \frac{\alphamax^5 }{\alphamin^6 {|f_1|}}  
\right) \Delta R^y
\label{eqdeltag}
\\
\EE \min_{c_2 \in \RR} \sup_{\om \in [0,1)} \left|\widehat\calR (\om) - \calR\left(\om-\frac{c_2}{2\pi}\right)\right|
&\le 
\frac{2}{\lam_s(F)} \Delta F
\label{eqdeltaR}
\end{align}
where 
$$\Delta F 
= \left[
\frac{9}{\alphamin^2} 
+ 
 \frac{12\alphamax^2\gammamax^2  \sigma^2_{\max}(A)}{\alphamin^3}
\left(
\frac{3(\|g\|^2+N \alphamax^2)}{2\alphamin\|g\|^2 f_0 }
+
 \frac{144 N^2\alphamax^5 }{\alphamin^6 {|f_1|}}  
\right)
\right] \Delta R^y.$$

%
\end{theorem}

\begin{remark}
The expectations in \eqref{eqdeltag} and \eqref{eqdeltaR} are taken over random source amplitudes and random noises. 
Our estimates suggest that, the partial algebraic method is more stable in the cases where 
1) the noise level $\sigma$ is small and the number of snapshots $L$ is large; 
2) $\|g\| , f_0, |f_1|$ are large;
3) the calibration parameters $g$ have a small dynamic range such that $\alphamax/\alphamin \approx 1$; 
4) the minimal calibration amplitude $\alphamin$ is large;  
5) source amplitudes have a small dynamic range such that $\gammamax/\gammamin \approx 1$;
6) frequencies are well separated such that $\smax(A)/\smin(A) \approx 1$.
\end{remark}

\begin{remark}
\label{remarkconstant}
 Notice that $\EE \|x(t)\| = \sqrt{\sum_{j=1}^s \gamma_j^2}$ and $\EE \|e(t)\|= \sigma\sqrt N$. Suppose $\|x(t)\|$ and $\|e(t)\|$ concentrate around $\EE\|x(t)\|$ and $\EE \|e(t)\|$ respectively. In the case that the true frequencies in $\calS$ are separated by $q >1/N$, discrete Ingham inequalities \cite[Theorem 2]{LFMUSIC} guarantee that $r_1(q,N) \le  \smin(A) \le \smax(A) \le r_2(q,N)$ for some positive constants $r_1,r_2$ depending on $q,N$, which implies $\lam_1(F) \le \gammamax^2 r_2^2$ and $\lam_s(F) \ge \gammamin^2 r_1^2$.
Then, when $L$ is sufficiently large, we have 
\begin{equation}
\label{eqDeltaRyB}
\Delta R^y \le 
O\left(\frac{B_1+B_2\max(\sigma,\sigma^2)}{\sqrt{L}}\right)
\end{equation}
for some positive constants $B_1,B_2$ depending on $g,\{\gamma_j\}_{j=1}^s,\calS,N ,s$. 
Therefore
\begin{align*}
\EE\min_{c_0>0,c_1,c_2\in \RR}
\max_n \|c_0 \hg_n -e^{i(c_1+nc_2)} g_n\|_\infty  
&\le O\left(\frac{C_1+ C_2\max(\sigma,\sigma^2)}{\sqrt{L}}\right)
\\
\EE \min_{c_2 \in \RR} \sup_{\om \in [0,1)} \left|\widehat\calR (\om) - \calR\left(\om-\frac{c_2}{2\pi}\right)\right| 
&\le O\left(\frac{C_3+C_4 \max(\sigma,\sigma^2)}{\sqrt{L}}\right)
\end{align*}
for some positive constants $C_1,C_2, C_3,C_4$ depending on $g,\{\gamma_j\}_{j=1}^s,\calS,N ,s$. 
In particular, if $\alphamax$, $\alphamin$, $\gammamax$, $\gammamin \approx 1$, and the frequencies are separated above $1/N$, then  $B_1 \sim \sqrt{s\log 4s}$, $B_2 \sim \sqrt{N\log(2N)}$, $C_1,C_3 \sim N^2\sqrt{s\log 4s}$ and $C_2, C_4 \sim N^2\sqrt{N\log(2N)}$.
\end{remark}

\begin{remark}
 In Theorem \ref{thmagstability}, the expression in \eqref{eqDeltaRy} may appear intimidating. However, it simply results from Bernstein inequalities \cite{Tropp}, based on which we estimate  the deviation of the sampled covariance matrix $\widetilde R^y_e$ from the covariance matrix $R^y_e$. Notice that 
\begin{align}
\|R^y_e - \widetilde R^y_e\|
& \le  \| G A (R^x - \widetilde R^x) A^* G^* + GA (R^{xe}-\wtR^{xe}) + (R^{ex}-\wtR^{ex}) A^* G^* + R^e - \wtR^e\| 
\nonumber
\\
&\le
\smaxsq(G) \smaxsq(A) \|R^x - \widetilde R^x\| + \smax(G)\smax(A) \|R^{xe}-\wtR^{xe}\| 
\label{eqber1}
\\
& \quad + \smax(G)\smax(A) \|R^{ex}-\wtR^{ex}\|  +  \|R^e - \wtR^e\|.
\label{eqber2}
\end{align}
Applying Bernstein inequalities gives rise to \eqref{eqDeltaRy} where the three terms correspond to upper bounds of \eqref{eqber1} and \eqref{eqber2}.
\end{remark}

\begin{remark}
By using Bernstein inequalities, we require that $\|x(t)\|$ and $\|e(t)\|$ are almost surely bounded and $\max_{t\in \Gamma} \| x(t)\|$ and $ \max_{t\in \Gamma} \|e(t)\|$ appear in the upper bound. 
This result can be generalized to the case where the entries in $x(t)$ and $e(t)$ are independent sub-gaussian random variables (so we can drop the boundedness condition) by using theorem 4.7.1 in \cite{rv2018}. 
Then \eqref{eqDeltaRy} becomes
\begin{align}
\Delta R^y 
:=&
C 
\bigg[ \frac{\alpha_{\max}^2\sigma^2_{\max}(A) \gammamax^3}{\gammamin^2} \left(\sqrt{\frac{s}{L}} +\frac{s}{L}\right)
\nn
+ 4\alpha_{\max}\smax(A)\gammamax  Ns\sqrt{Ns}\sigma
 \left(\frac{1}{\sqrt{L}}+\frac{1}{L}\right)
 \nn
\\
&
 +2 \sigma \left(\sqrt{\frac{N}{L}}+\frac{N}{L}\right)\bigg],
\end{align}
and other results hold similarly.
\end{remark}



 A sensitivity analysis of the full algebraic method to the number of snapshots can be found in \cite{LiEr}. Assuming the problem of phase wrapping in the full algebraic method is resolved, Li and Er \cite{LiEr} split the reconstruction errors of the calibration amplitudes and phases to a bias term and a variance term. They claim that the bias is nonzero, and the variance of the calibration phases is $O(1/\sqrt L)$ where $L$ is the number of snapshots. Below we point out some differences between Theorem \ref{thmagstability} and the analysis in \cite{LiEr}.
\begin{enumerate}
\item In \cite{LiEr} the authors did not give an explicit bound on the bias but claimed it is non-zero. In this case the total reconstruction error for the calibration phases does not approach $0$ as $L \rightarrow \infty$. In comparison, we show that the reconstruction error of the calibration parameters and the perturbation of the noise-space correlation function in MUSIC converge to $0$ as $L \rightarrow \infty$ in Theorem \ref{thmagstability}.

\item We present a sensitivity analysis of the partial algebraic method to both the number of snapshots and noise, while the sensitivity to noise is not addressed in \cite{LiEr}.

\item The upper bounds in Theorem \ref{thmagstability} are explicitly given in terms of $g$, $\{\gamma_j\}_{j=1}^s$, $N$, $s$, $\smax(A)$ and $\smin(A)$. When the underlying frequencies are separated by $q>1/N$, we can further bound $\smax(A)$ and $\smin(A)$ in terms of $q$ and $N$ by discrete Ingham inequalities \cite{LFMUSIC}. In comparison, all bounds in \cite{LiEr} are implicit in the sense that the bias is defined but not estimated, and the variance is expressed in terms of the trace of certain matrices that are not explicitly given.

\item One needs to perform standard spectral estimation after calibration parameters are recovered. Theorem \ref{thmagstability} includes a sensitivity analysis of the MUSIC algorithm, which is not addressed in \cite{LiEr}.

\end{enumerate}

\section{Optimization approach}
\label{sec:optimization}

As discussed in Section \ref{SecAlgebraicFull}, it is nontrivial to make use of all entries in the covariance matrix in algebraic methods.
Instead we now propose an optimization approach which takes advantage of all measurements.

Suppose $\hR^y$ is an estimate of $R^y$.
According to Lemma \ref{lemma2}, we can recover exact calibration parameters $g$ and the vector $f$ defined in \eqref{eqf}, by solving the following optimization problem:
\beq
\label{optproblem}
\min_{\bg ,\bff \in \CC^{N}} \calL(\bg,\bff):=\left\| \diag( \bg) \calT(\bff) \diag(\bar{ \bg}) - \hR^y\right\|_F^2.
\eeq
Here we use boldface letters $\bg,\bff$ to denote variables in optimization and $g,f$ to denote the ground truth.
With infinite snapshots of noiseless measurements, the covariance matrix is exactly known so that $\hR^y = R^y$, and Lemma \ref{lemma2} implies that the global minimizer of \eqref{optproblem} is the ground truth up to a trivial ambiguity. If finite snapshots of noisy measurements are taken, then we run Steps 1 - 4 in Algorithm \ref{algorithm1} to obtain $\hR^y$ as an approximation to $R^y$.

As pointed out in Lemma \ref{lemma2}, if $(\bg,\bff)$ is a solution to \eqref{optproblem}, then so is $(c_0 \bg, c_0^{-2}\bff)$ for any $c_0 \neq 0$. In order to guarantee numerical stability, we avoid the case that  $\|\bg\| \rightarrow 0$ and $\|\bff\| \rightarrow \infty$ (or vice versa) by adding a penalty to the objective function. Let $ n_0 := \|g\|^2 \|f\|$ which can be estimated from $\hR^y$ based on the following lemma (see Appendix \ref{prooflemman0} for the proof):
\begin{lemma}
\label{lemman0}
Let $R^y$ be defined in \eqref{eqRy}. Then 
\beq
\label{eqn0}
n_0 = \left(\sum_{n=0}^{N-1} R^y_{n,n}  \right)\sqrt{1 +  \frac{1}{N-k}\sum_{k=1}^{ N-1} \sum_{n=0}^{N-k-1} \frac{|R^y_{n+k,n}|^2}{R^y_{n+k,n+k} R^y_{n,n}}}.
\eeq
\end{lemma}

Let $\hn_0$ be an estimate of $n_0$ from \eqref{eqhatn0}. Theorem \ref{thmagstability} shows that $\|\hR^y - R^y\| \le \Delta R^y $ with $\Delta R^y$ given by \eqref{eqDeltaRy}. When the true frequencies are separated by $1/N$, \eqref{eqDeltaRyB} implies that 
$\Delta R^y \le O ( L^{-\frac 1 2} )$ and $\hn_0 \approx n_0$ when $L$ is sufficiently large.

 \begin{algorithm}[ht!]                      	
\caption{Optimization approach}          	
\label{algorithm2}		
\begin{algorithmic}[1]                    	
    \REQUIRE  Measurements $\{y_e(t), \ t \in \Gamma\}$  and sparsity $s$.
        \ENSURE Calibration parameters $\widehat{\bg}$ and recovered spectrum $\{\widehat{\om}_j\}_{j=1}^s$.
    \STATE Run Step 1-7 in Algorithm \ref{algorithm1} and compute 
    \begin{equation}
    \widehat{n}_0 = \sum_{n=0}^{N-1} \hR^y_{n,n}  \sqrt{1 +  \frac{1}{N-k}\sum_{k=1}^{ N-1} \sum_{n=0}^{N-k-1} \frac{|\hR^y_{n+k,n}|^2}{\hR^y_{n+k,n+k} \hR^y_{n,n}}}.
    \label{eqhatn0}
    \end{equation}
    \STATE {\bf Initialization:} 
    \begin{description}
     \item[(i)] Let $\bg^0 \leftarrow \hg$ and $F^0 \leftarrow \hF$ where $\hg$ and $\hF$ are from Steps 1-7 of Algorithm \ref{algorithm1}.
     \item[(ii)] Let $\bff^0 \in \CC^N$ such that $\bff^0_k = \frac{1}{N-k}\sum \diag(F^0,-k), k=0,\ldots,N-1$.
     \item[(iii)] Normalization: $\bg^0 \leftarrow \sqrt[4]{\hn_0}\frac{\bg^0}{\|\bg^0\|}$ and $\bff^0 \leftarrow \sqrt{\hn_0}\frac{\bff^0}{\|\bff^0\|} $.
    \end{description}
    \FOR{$k = 1,2,\ldots,$} 
    \STATE $\bg^k = \bg^{k-1} - \eta^k \nabla_{\bg} \tcalL(\bg^{k-1},\bff^{k-1})$.
   \STATE $\bff^k = \bff^{k-1} - \eta^k \nabla_{\bff} \tcalL(\bg^{k-1},\bff^{k-1})$.    \ENDFOR
  \STATE {Output of gradient descent:} $\widehat\bg$ and $\widehat\bff$.
     \STATE Apply MUSIC to $\calT(\widehat\bff)$ to obtain the spectrum $\{\widehat\om_j\}_{j=1}^s$. 
     \commentout{
    \begin{description}
     \item[i)] Compute the eigenvalue decomposition: $$\calT(\widehat\bff) = [V_1 \ V_2] \diag(\lam_1(\calT(\widehat\bff)),\ldots,\lam_s(\calT(\widehat\bff)),\ldots) [V_1 \ V_2]^*$$ where $V_1 \in \CC^{N \times s}$, and $\lam_1(\calT(\widehat\bff)) \ge \lam_2(\calT(\widehat\bff))\ge \ldots$.
      \item[ii)] Compute the imaging function $\widehat{\calJ}(\om) =\frac{ \|\phi(\om)\|}{\|V_2^*\phi(\om)\|}$ where $\phi(\om) = [1 \ e^{2\pi i \om} \ \ldots e^{2\pi i(N-1) \om}]^T$
      \item[iii)] Return the spectrum $\{\widehat\om_j\}_{j=1}^s$ corresponding to the $s$ largest local maxima of $\widehat{\calJ}(\om)$
     \end{description}
           }
\end{algorithmic}
\end{algorithm}

Consider the following bounded set:
\beq
\label{eqcalNn0}
\calN_{\hn_0} = \{(\bg,\bff): \|\bg\|^2 \le 2\sqrt{\hn_0}, \|\bff\| \le 2\sqrt{\hn_0}\}.
\eeq
We pick an initial point satisfying 
\beq
\label{optini}
(\bg^0,\bff^0): \|\bg^0\|^2 \le \sqrt{2 \hn_0}, \|\bff^0\| \le \sqrt{2 \hn_0}.
\eeq
through the partial algebraic method. The solution from the partial algebraic method has a scaling ambiguity, so we
simply normalize it to guarantee \eqref{optini}.
In order to ensure all the iterates remain in $\calN_{\hn_0}$, we minimize the following regularized function:
\beq
\label{optproblemre}
\min_{\bg,\bff \in \CC^N}\tcalL(\bg,\bff) := \calL(\bg,\bff) + \calG(\bg,\bff)
\eeq
where $\calL(\bg,\bff)$ is defined in \eqref{optproblem} and $\calG(\bg,\bff)$ is a penalty function of the form
$$\calG(\bg,\bff) = \rho  \left[ \calG_0\left(\frac{\|\bff\|^2}{2 \hn_0} \right)  
+   \calG_0\left(\frac{\|\bg\|^2}{\sqrt{2 \hn_0}} \right)   \right]$$
where $\calG_0(z) = (\max(z-1,0))^2$ and $\rho \ge {(\sqrt 2 -1)^{-2}} \left(3\hn_0 +\|R^y -\hR^y\|_F \right)$. 
When the exact frequencies are separated by $1/N$, we have $\|\hR^y-  R^y\| = O[ L^{-\frac 1 2} (B_1+B_2\max(\sigma,\sigma^2)) ]$. 
It follows that
$
\|\hR^y - R^y\|_F \le \sqrt{N} \|\hR^y-R^y\| \le \sqrt{N}\Delta R^y \rightarrow 0$ as $L \rightarrow \infty$. We therefore take $\rho \ge 3(\sqrt 2 -1)^{-2} \hn_0$ when $L$ is sufficiently large.

The objective function in \eqref{optproblemre} is continuously differentiable but non-convex. We choose an initial point satisfying \eqref{optini} by the partial algebraic method and
solve \eqref{optproblemre} by gradient descent where the derivative can be interpreted as a Wirtinger derivative \footnote{Let $z = x + i y$ and $h(z) = h(x,y) = u(x,y)+ i v(x,y)$. The Wirtinger derivatives and gradient of $h$ are
$$\frac{\partial h}{\partial z} := \frac{1}{2} \left(\frac{\partial h}{\partial x} - i \frac{\partial h}{\partial y} \right), \quad
\nabla_z h := \frac{\partial h}{\partial \bar{z}} := \frac{1}{2} \left(\frac{\partial h}{\partial x} + i \frac{\partial h}{\partial y} \right).$$}.
The Wirtinger gradient of $\tcalL$ is given by
$$ \nabla \tcalL 
 =\left[\nabla_{\bg} \tcalL \  \ \nabla_{\bff} \tcalL  \right]^T
=\left[\nabla_{\bg} \calL + \nabla_{\bg} \calG \quad \nabla_{\bff} \calL + \nabla_{\bff} \calG  \right]^T 
$$
with 
\begin{align}
\nabla_{\bg} \calL & =2 \diag \Big[
\overline{
\calT(\bff)^*  \diag(\bar{\bg}) 
\Big( \diag(\bg)\calT(\bff)\diag(\bar \bg) - \hR^y 
\Big)}
\Big],
\label{eqdg}
\\ 
\nabla_{\bff} \calL
& = \calT^a \Big[
\overline{
\diag(\bar \bg)  
\Big(
\diag(\bg) \calT(\bff) \diag(\bar \bg) - \hR^y
\Big)
\diag(\bg)}
\Big],  
\label{eqdf}
\\
\nabla_{\bg} \calG 
&= \frac{\rho}{\sqrt{2 \hn_0}} \calG_0'\left(\frac{\|\bg\|^2}{\sqrt{2 \hn_0}} \right)  \bg,
\nonumber
\\
\nabla_{\bff} \calG 
&= \frac{\rho}{2\hn_0} \calG'_0 \left( \frac{\|\bff\|^2}{2\hn_0} \right) \bff, 
\nonumber
\end{align}
and $\calG'_0 (z) = 2 \max(z-1,0)$. The operator $\calT^a: \CC^{N \times N} \rightarrow \CC^N$ is defined as 
$$\calT^a: \CC^{N \times N} \rightarrow \CC^N:
\
\calT^a(X) = 
\begin{bmatrix}
\sum \left(\diag(X)+ \diag(\bar X)\right) \\
\sum\left( \diag(X,1) +\diag(\bar X , -1)\right) \\
\vdots \\
\sum \left( \diag(X,N-1) + \diag(\bar X,-(N-1))\right) \\
\end{bmatrix}.$$
One can verify that $\calT^a$ satisfies
$$
\frac{\partial}{\partial \bff} \Big( \left\langle \calT(\bff), X \right\rangle + \left\langle X , \calT(\bff) \right\rangle \Big) 
= \calT^a(X), \ \forall X \in \CC^{N \times N},
$$
and therefore
\begin{align*}
\frac{\partial \calL}{\partial \bff} 
& = \frac{\partial}{\partial \bff}
\Big( 
\left\langle \calT(\bff), \diag(\bar\bg) \left(\diag(\bg)\calT(\bff)\diag(\bar\bg) -\hR^y \right)\diag(\bg)\right\rangle
\\ 
&+
\left\langle  \diag(\bar\bg) \left(\diag(\bg)\calT(\bff)\diag(\bar\bg) -\hR^y \right)\diag(\bg), \calT(\bff)\right\rangle
\Big)
\\
& = \calT^a \Big[
\diag(\bar \bg)  
\Big(
\diag(\bg) \calT(\bff) \diag(\bar \bg) - \hR^y
\Big)
\diag(\bg)
\Big]
\end{align*}
which gives rise to \eqref{eqdf}.
%

%
\commentout{ 
{\color{blue} If we only consider the lower triangular part in \eqref{optproblem},
\begin{align*}
\frac{ \partial \calL}{\partial g_n}
&=\sum_{i=0}^{n-1} \barg_i f_{n-i}( g_i \barg _n {\barf_{n-i}} -\bar{R}^y_{n,i})+
2 \barg_n f_0 (g_n \barg_n f_0 -R^y_{n,n})
+
\sum_{j=1}^{N-n-1}\barg_{n+j}\barf_j( g_{n+j} \barg_n f_j-R^y_{n+j,n}) 
\\
 \frac{ \partial \calL}{\partial f_n} 
 &=\sum_{i=0}^{N-n-1}g_{n+i}\barg_i(\barg_{n+i} g_i\barf_n -\bar{R}^y_{n+i,i}), ~~~n=0,\cdots,N-1.
 \end{align*}
}  
}

Our optimization approach for sensor calibration is summarized in Algorithm \ref{algorithm2}. In the next theorem we prove that the gradient descent in Steps 3-6 of Algorithm \ref{algorithm2} converges to a critical point of \eqref{optproblemre}.
\begin{theorem}
\label{thmwirtinger}
Let $(g,f)$ be the ground truth. Let $\hR^y$ be an estimate of $R^y$. Assume that the initial point $(\bg^0,\bff^0)$ satisfies $\|\bg^0\| \le \sqrt[4]{2\hn_0}$ and $\|\bff^0\| \le \sqrt{2\hn_0}$, and $\rho \ge  {(\sqrt 2 -1)^{-2}}(3\hn_0 +\|R^y -\hR^y\|_F)$. Then running Algorithm \ref{algorithm2} with step size 
\beq
\label{eqetak}
\eta^k \le 2/C_{\rm Lip} 
\eeq
where
$$C_{\rm Lip} \le 166\hn_0 \max(\sqrt{\hn_0},\sqrt[4]{\hn_0})+8\hn_0 + 16\max(\sqrt{\hn_0},\sqrt[4]{\hn_0})\|R^y -\hR^y\|_F + \frac{12\rho}{\min(\hn_0,\sqrt{\hn_0})}$$
gives rise to a sequence $(\bg^k,\bff^k) \in \calN_{\hn_0}$, and 
$$\|\nabla \tcalL(\bg^k,\bff^k)\| \rightarrow 0, \text{ as } k\rightarrow \infty.$$

\end{theorem}

Theorem \ref{thmwirtinger} (see Appendix \ref{appthmwirtinger} for the proof) shows that Wirtinger gradient descent converges to a critical point of \eqref{optproblemre}. Our numerical experiments suggest that this point indeed provides a good approximation of the ground truth up to a trivial ambiguity.

\commentout{
\section{Other thoughts and possibilities}

\subsection{Formulation 2: sensing matrix given}
Consider the linear system
\beq
\label{eqprob2}
y = GAx
\eeq
where $x \in \CC^M,$ $A \in \CC^{N \times M}$, $y \in \CC^{N}$ and $D =\diag(g)$.
Here $A$ is a known sensing matrix.
Our goal is to recover $x$ and $G$ from $y$.

\subsubsection{Uniqueness results in \cite{LiBresler}}

$\#$ snapshot: $L$, $X = [x_1 \ x_2 \ \ldots x_L]$, generic $D,A,X$

\begin{itemize}
\item Result 1: $D$ and $X$ are uniquely determined up to a scaling if $N > M$ and $M \ge L \ge (N-1)/(N-M)$.

\item Result 2: Suppose $X$ has $s$ nonzero rows. $D$ and $X$ are uniquely determined up to a scaling if $N > 2s$ and $ s \ge L \ge (N-1)/(N-2s)$.
\end{itemize}

\subsection{Kailath's method on blind deconvolution}
Consider $$y(i) = H x(i).$$
Let $H = DA$ and suppose $x(i), i=1,\ldots,$ is a zero mean stationary process with autocorrelation function
$$R^x(k) = \EE x(i)x^*(i-k)$$
of the form: $R^x (k) = J^k$ when $k \ge 0$ and $R^x(k) = (J^*)^{|k|}$ when $k<0$ where 
$$J = \begin{bmatrix}
0 & 0 & \ldots & 0 & 0 \\
1 & 0 & \ldots & 0 & 0 \\
0 & 1 & \ldots & 0 & 0 \\
 &  & \ddots &  &  \\
0 & 0 & \ldots & 1 & 0 \\
\end{bmatrix} \in \RR^{s \times s}.$$
Kailath's method \cite{Kailath94} is to recover $H$ from $R^y(0) = H H^*$ and $R^y(1) = H J H^*$.

\subsection{ESPRIT}
Let  $D_1 = D(1:M-1,1:M-1)$, $D_2= D(2:M,2:M)$, and $A_1 = A(1:M-1,:)$,. Consider
\begin{align*}
y_1(t) & = y(1:M-1)= D_1 A_1 x(t) \\
y_2(t) & = y(2:M) = D_2 A_1 \Phi x(t).
\end{align*}
where $\Phi_{jj} = e^{2\pi i \om_j}, \ j=1,\ldots,s$ is a diagonal matrix.
Then the covariance matrices of $y_1(t)$ and $y_2(t)$ are
\begin{align*}
R^{y_1} & = \EE y_1(t) y_1^*(t) =D_1 A_1 \EE x(t)x^*(t) A_1^* D_1^* ,
\\
R^{y_2} & = \EE y_2(t) y_2^*(t) = D_2 A_1 \Phi \EE x(t)x^*(t) \Phi^* A_1^* D_2^*.
\end{align*}
\textcolor{red}{Difficulty: to produce the shifting matrix $J$.}
}

\section{Numerical experiments}
\label{sec:numerics}

We perform  systematic numerical simulations to compare the performance of existing methods for the sensor calibration problem modeled by \eqref{eqprob1}. In our simulations, $\supp$ contains $s$ frequencies located on the continuous domain $[0,1)$.  Theorem \ref{thmagstability} shows that the problem is more challenging when the dynamic ranges of $x$ and $g$ increase. We denote the dynamic range of $\gamma$ and $g$ by $\DR_\gamma>0$ and $\DR_g>0$ respectively, and let $\gamma_i = (\EE x_i^2(t))^{\frac 1 2} \in [1,\DR_\gamma], i=1,\ldots,s$ and $|g_i| \in [1,\DR_g], i = 0,\ldots,N-1$. The phases of $x_j(t)$ are randomly chosen from $[0,2\pi)$ to guarantee that $\EE x(t)x^*(t) = \diag(\{\gamma_i^2\}_{i=1}^s)$. We add i.i.d. Gaussian noise to the measurements such that $y_e(t) = y(t)+e(t)$ with $e(t) \sim \mathcal{N}(0,\sigma^2 I_N)$. 

Suppose we take $L$ snapshots of independent measurements, i.e., $\{y_e(t): t =1,\ldots,L\}$, and form the empirical covariance matrix $\wtR_e^y$.
%
We assume $s$ is known and denote the support of the recovered frequencies by $\widehat\supp = \{\widehat\om_j\}_{j=1}^s$. Due to the discrete set-up of sensors, we assume periodicity of the frequency domain $[0,1)$ on which the distance between two frequencies $d(\om_j,\om_l)$ is understood as the wrap-around distance on the torus.
Frequency support error is measured by the Hausdorff distance between $\supp$ and $\widehat\supp$ up to a translation:
\beq
\label{eqSuppError}
\text{SuppError} =d(\supp,\widehat\supp) :=\min_{c_2 \in [0,2\pi)} \max \left(\max_{\widehat\om \in \widehat\supp} \min_{\om\in \supp} d\left(\widehat\om+\frac{c_2}{2\pi},\om\right) ,\  \max_{\om \in \supp}\min_{\widehat \om \in \widehat\supp} d\left(\widehat\om+\frac{c_2}{2\pi},\om\right)\right).
\eeq
Let $c_2^*$ be the minimizer in \eqref{eqSuppError}. In the noiseless case, we expect  the recovered calibration parameters to be of the form $\widehat g_n = c_0 e^{i c_1 }e^{i n c_2^*}g_n$ for some $c_0>0$ and $c_1 \in [0,2\pi)$. Let $\tilde g_n = \widehat g_n e^{- i nc_2^*}$ and $C^* = {\rm argmin}_{C} \sum_{n=0}^{N-1} |\tilde g_n - C g_n|^2$. We measure the relative calibration error for the $n$-th sensor and the average relative calibration error as 
$$\text{CalError}_n = \frac{|\tilde g_n - C^* g_n|}{|g_n|}, 
\quad
\text{CalError} = \frac{1}{N} \sum_{n=0}^{N-1}\text{CalError}_n. 
$$

We test the following methods:

\begin{itemize}
\item the partial algebraic method in Algorithm \ref{algorithm1};

\item  the optimization approach in Algorithm \ref{algorithm2}: 
In practice we choose the step length $\eta^k$ according to the backtracking line search in Algorithm \ref{algorithmline} \cite[Algorithm 3.1]{SWright}. This backtracking approach ensures that the selected step length $\eta_k$ is short enough to guarantee a sufficient decrease of $\tcalL$ but not too short. The latter claim holds since the accepted step length $\eta_k$ is within a factor $\theta$ of the previous trial value $\eta^k/\theta$, which was rejected for violating the sufficient decrease of $\tcalL$, that is, for being too large. In our implementation, we set $\bar\eta = \tcalL(\bg^k,\bff^k)/\|\nabla \tcalL(\bg^k,\bff^k)\|$, $\theta = 0.5, c=0.5$ and terminate gradient descent while $\eta^k < 10^{-4}$;

\begin{algorithm}[ht]                      	
\caption{Backtracking line search}          	
\label{algorithmline}		
\begin{algorithmic}[1]                    	
    \STATE Choose $\bar\eta>0, \theta \in (0,1), c \in (0,1)$; Set $\eta \leftarrow \bar\eta$; Denote $\bz^k = (\bg^k,\bff^k)$ and $\bp_k =- \nabla \tcalL(\bg^k,\bff^k)$.
    \REPEAT  
    \STATE  $\eta \leftarrow \theta \eta$
    \UNTIL $\calL(\bz_k+\eta \bp_k) \le \calL(\bz_k) - c\eta\|\bp_k\|^2$.
    \RETURN $\eta^k = \eta$. 
\end{algorithmic}
\end{algorithm}

\item an alternating algorithm proposed by Friedlander and Weiss \cite{FriedlanderWeiss}. This algorithm is based on a two-step procedure. First, one assumes that the calibration parameters are known, and estimates the frequencies with the MUSIC algorithm. Given the recovered frequencies, one minimizes the squared sum of the noise-space correlation functions evaluated at the recovered frequencies over all calibration parameters. We {choose an initial point using the partial algebraic method} and terminate the iterations when the squared sum of the noise-space correlation functions evaluated at the recovered frequencies decreases by $10^{-4}$ or less.

\end{itemize}

\subsection{Partial algebraic method and optimization approach}

We expect the optimization approach to outperform the partial algebraic method in almost all cases since all measurements are used. To illustrate this, we perform reconstructions on $20$ frequencies separated by $2/N$. We set $\DR_\gamma = \DR_g = 2$ and $\sigma = 0.5$. We apply both methods to the same set of measurements. In Figure \ref{FigImaging}, the imaging functions in the MUSIC algorithm are displayed for the partial algebraic method and the optimization approach, respectively. Both techniques succeed as imaging functions peak around the true frequencies. However, the optimization approach yields peaks that are higher and sharper, and the support error is smaller. 

\begin{figure}[hthp]
\centering
\subfigure[Imaging function in the partial algebraic method]{
\includegraphics[width=8cm]{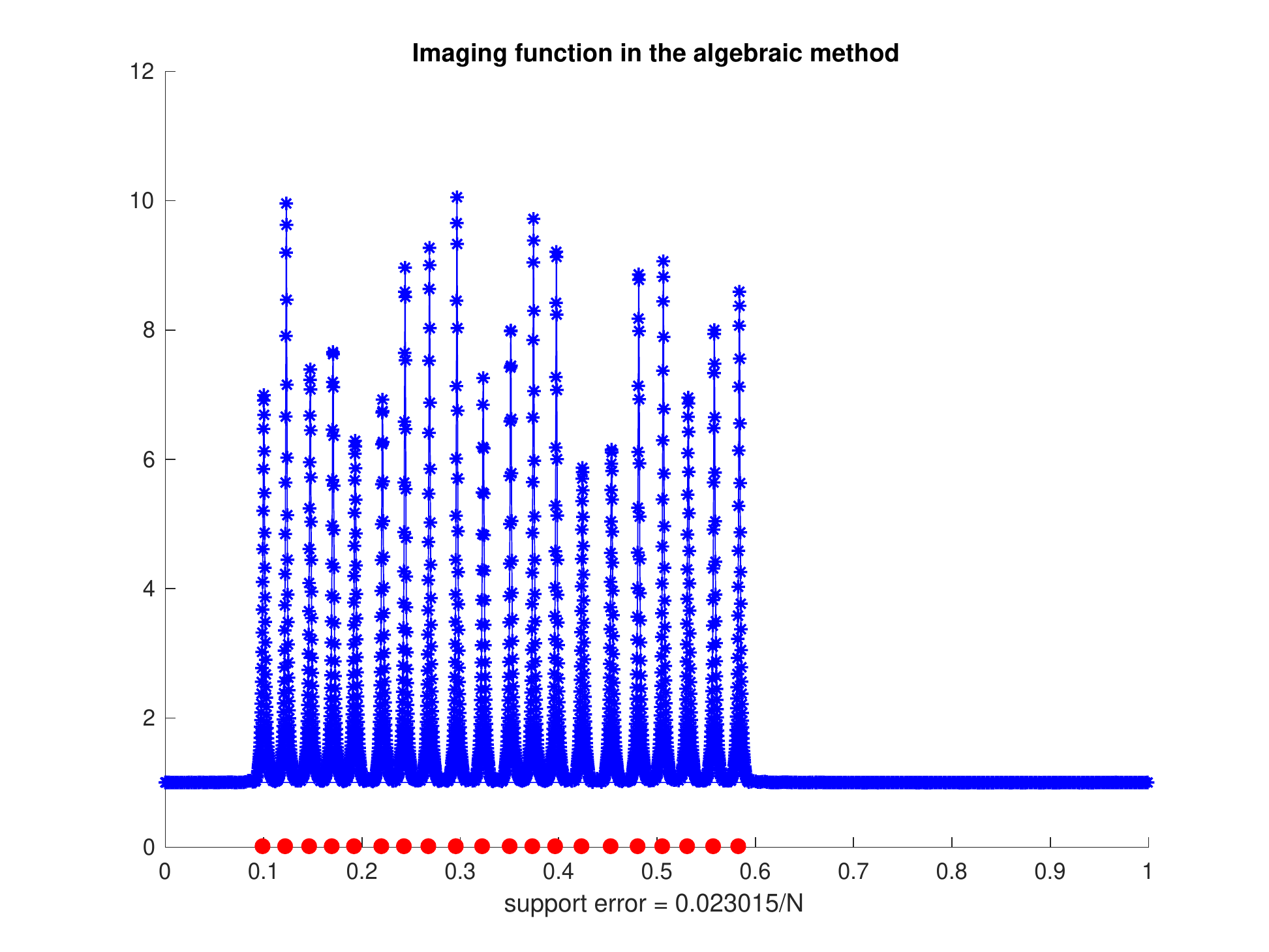}
}
\hspace{-1cm}
\subfigure[Imaging function in the optimization approach]{
\includegraphics[width=8cm]{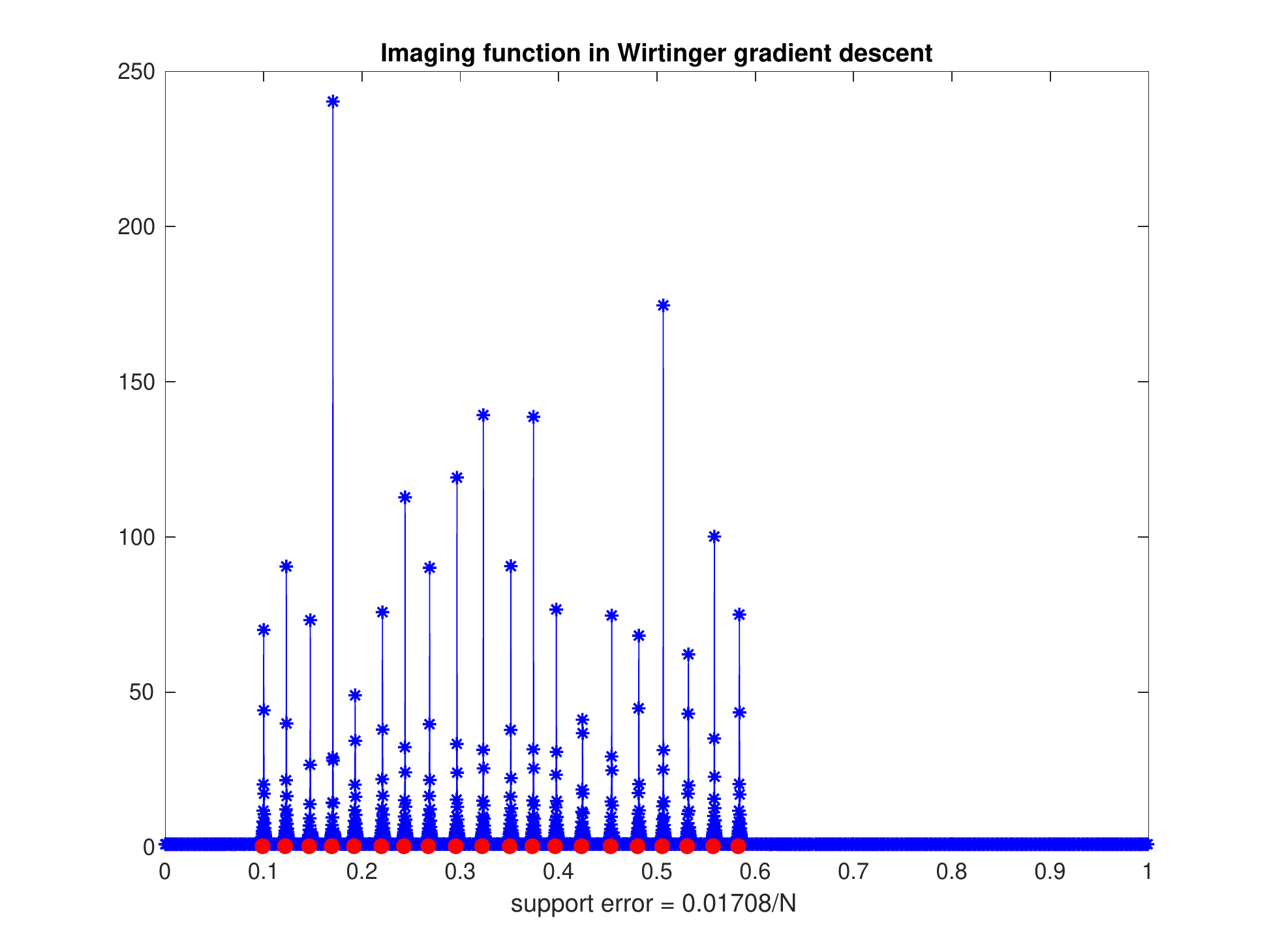}
}
\caption{Imaging functions (after a proper translation) in the MUSIC algorithm for the partial algebraic method (a) and the optimization approach (b). Red dots represent the locations of true frequencies. The two methods are applied on the same set of measurements generated by $20$ frequencies separated by $2/N$, when $\DR_\gamma = \DR_g = 2$, $L=500$ and $\sigma = 0.5$.}
\label{FigImaging}
\end{figure}

\commentout{
\subsection{Global convergence of gradient descent}

In the optimization approach, we solve a non-convex optimization problem \eqref{optproblemre} via Wirtinger gradient descent. Theorem \ref{thmwirtinger} guarantees that the Wirtinger gradient descent asymptotically converges to a critical point. We observe that this point is indeed a global minimum of $\calL$ in numerical experiments.  As an example, we perform reconstructions on $20$ frequencies separated by $2/N$ when the noise level $\sigma =0 $ and $0.5$ respectively. We set $\DR_\gamma = \DR_g = 2$. Figure \ref{FigWirtingerConvergence} displays $\log_{10}\calL(\bg^k,\bff^k)$ versus $k$ in $5$ independent experiments. The horizontal line represents $\log_{10}\calL(g,f)$ where $(g,f)$ is the background truth. In all experiments $\calL(\bg^k,\bff^k)$ dropped below $\calL(g,f)$ within $50$ iterations.

\begin{figure}[hthp]
\centering
\subfigure[$\log_{10}\calL(\bg^k,\bff^k)$ versus $k$ when $\sigma = 0$]{
\includegraphics[width=8cm]{WirtingerConvergence/WirtingerConvergenceNoise0.eps}
}
\hspace{-1cm}
\subfigure[$\log_{10}\calL(\bg^k,\bff^k)$ versus $k$ when $\sigma = 0.5$]{
\includegraphics[width=8cm]{WirtingerConvergence/WirtingerConvergenceNoise50.eps}
}
\caption{This figure displays $\log_{10}\calL(\bg^k,\bff^k)$ versus $k$ in $5$ independent experiments in the optimization approach, when the noise level $\sigma = 0$ and $0.5$ respectively.
 The horizontal line represents $\log_{10}\calL(g,f)$ where $(g,f)$ is the background truth. In all experiments $\calL(\bg^k,\bff^k)$ dropped below $\calL(g,f)$ within $50$ iterations.
}
\label{FigWirtingerConvergence}
\end{figure}
}

\subsection{Sensitivity to the number of snapshots}

The performance of all algorithms improves as the number of snapshots $L$ increases. We prove in Theorem \ref{thmagstability} that, for the partial algebraic method, when the underlying frequencies are separated by $1/N$ or above, the reconstruction error of calibration parameters decays like $O(1/\sqrt{L})$. In order to verify this result, we perform reconstructions on $20$ frequencies separated by $2/N$ when $L$ increases from $30$ to $10^4$. We set $\DR_\gamma = \DR_g =2$, and let the noise level be $\sigma = 0,0.5,1,2$. 
Figure \ref{FigErrVersusL} displays the relative reconstruction error of calibration parameters and the success probability of support recovery in $100$ independent experiments versus $L$ in a logarithmic scale. The frequency support is successfully recovered if $d(\calS,\widehat\calS)\le 0.2/N$. We observe that, (1) the reconstruction errors of calibration parameters for the partial algebraic method and the optimization approach decay like $O(1/\sqrt{L})$ since the slopes in Figure \ref{FigErrVersusL} (a) are roughly $-0.5$; (2) {in terms of stability to the number of snapshots, the alternating algorithm in \cite{FriedlanderWeiss} works the best when $\sigma = 0$, but its performance degrades dramatically when noise exists. In the presence of noise, our optimization approach has the best performance, and the partial algebraic method is the second best performer.}

\begin{figure}[ht!]
\centering
\subfigure[Relative calibration error versus $L$]{
\includegraphics[width=8cm]{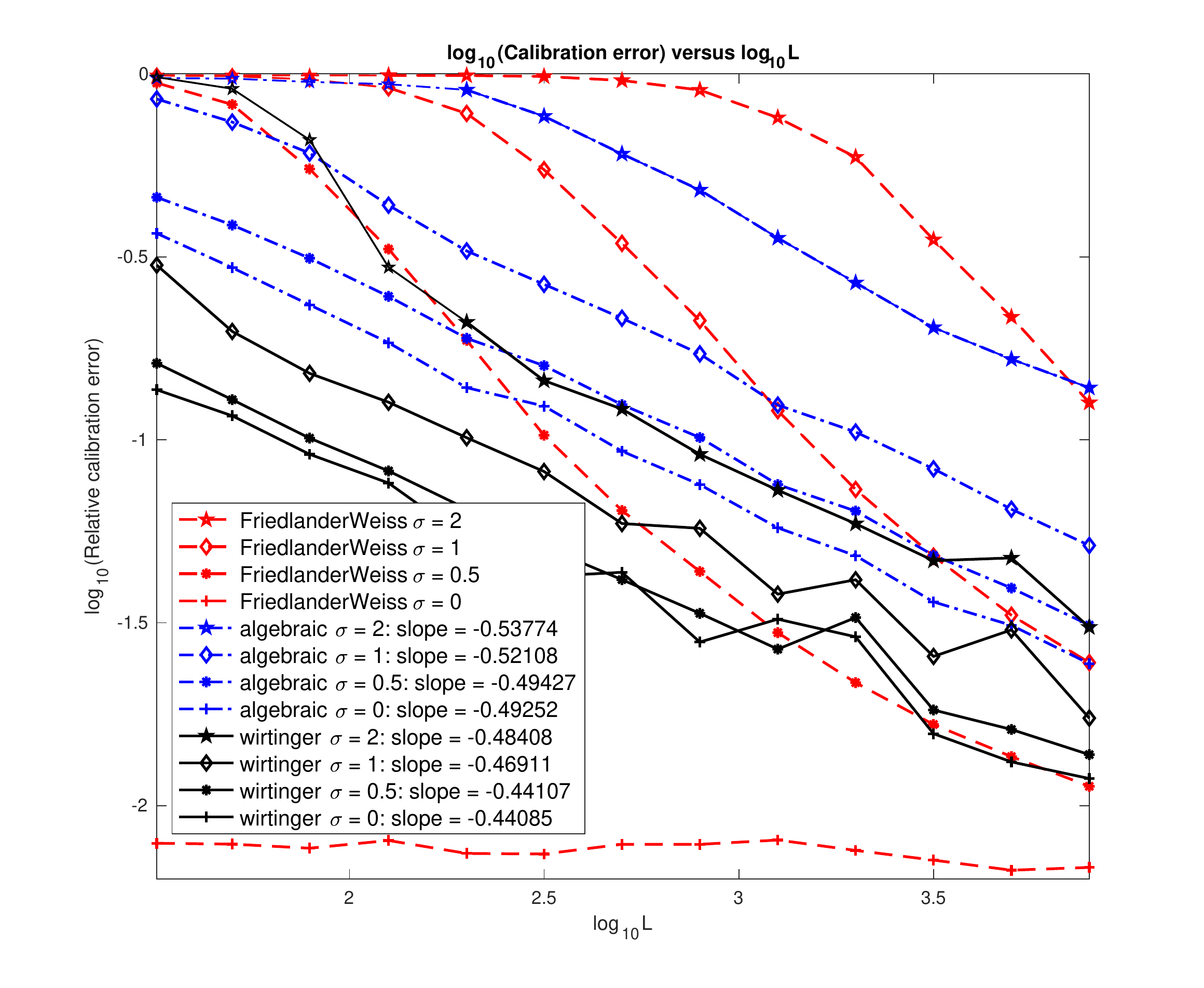}
}
\hspace{-1cm}
\subfigure[Frequency support success probability versus $L$]{
\includegraphics[width=8cm]{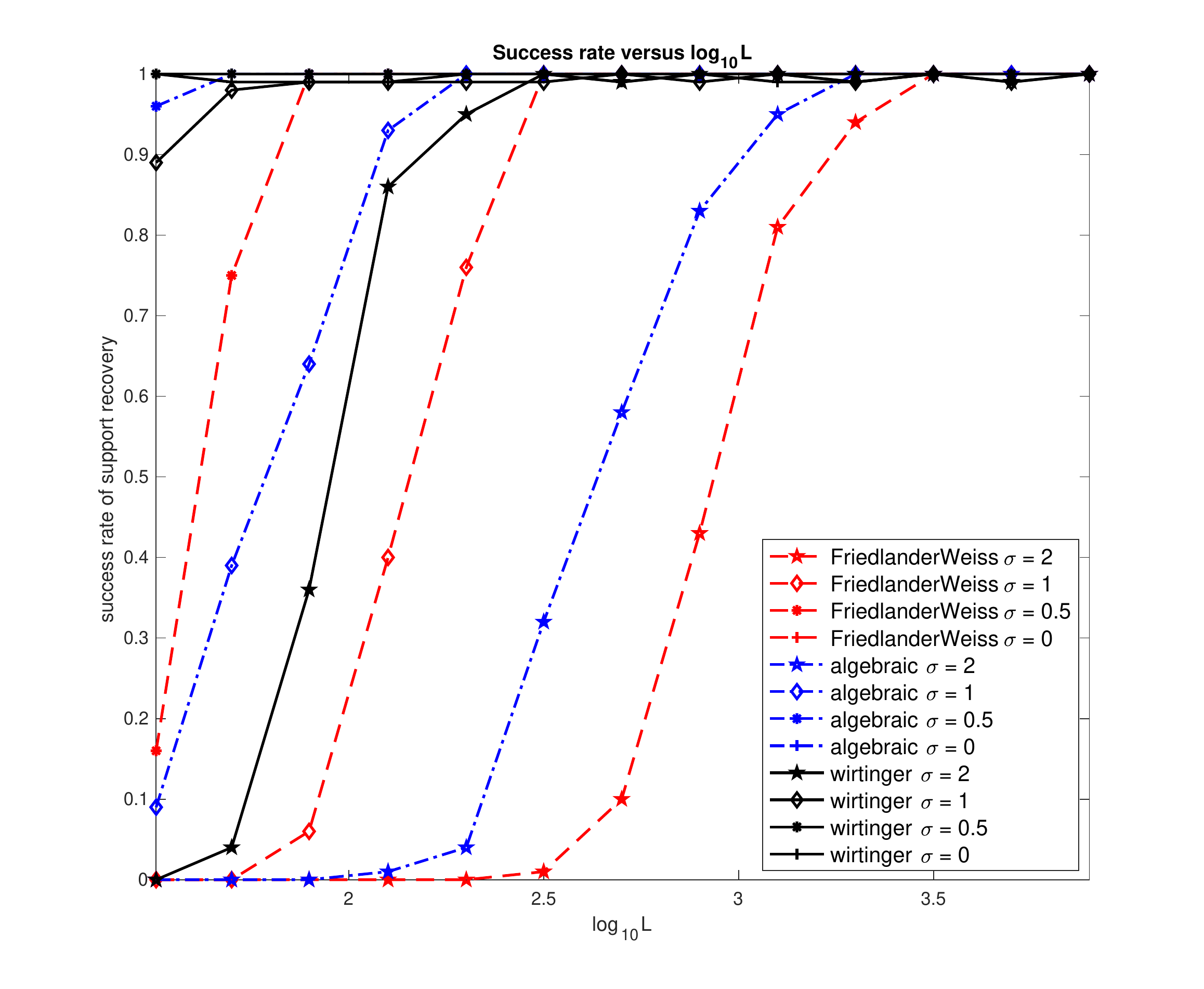}
}
\caption{(a) and (b) show the average relative reconstruction error of calibration parameters and the success probability of support recovery in $100$ independent experiments versus $L$ in a logarithmic scale. We choose $20$ frequencies separated by $2/N$, $\DR_\gamma = \DR_g =2$, and noise level $\sigma = 0,0.5,1,2$. 
}
\label{FigErrVersusL}
\end{figure}

\subsection{Sensitivity to noise}
To test the sensitivity of the various approaches to noise, we perform reconstructions on $20$ frequencies separated by $2/N$ when $\sigma$ increases from $10^{-1}$ to $10$. We set $\DR_\gamma = \DR_g =2$, and let $L = 500, 1000$ respectively. The frequency support is successfully recovered if $d(\calS,\widehat\calS)\le 0.2/N$. Figure \ref{FigErrVersusNoise} displays the average reconstruction error of calibration parameters and the success probability of support recovery in $100$ independent experiments versus $\sigma$ in a logarithmic scale. We observe that, (1) the reconstruction errors of calibration parameters for the partial algebraic method and the optimization approach increase like $O(\sigma)$ when $\log_{10}\sigma$ varies from $-0.5$ to $0.6$ since the slopes in Figure \ref{FigErrVersusNoise} (a) are roughly $1$; (2) {in terms of stability to noise, the alternating algorithm in \cite{FriedlanderWeiss} yields the smallest calibration error when $\sigma$ is small. As the noise level increases, our optimization approach becomes the best performer, while the partial algebraic method is the second best. }Notice that the reconstruction errors do not necessarily approach $0$ when $\sigma$ decreases to $0$ due to deviation of the empirical covariance matrix from the true covariance matrix caused by the finite number of snapshots.

\begin{figure}[ht!]
\centering
\subfigure[Relative calibration error versus $\sigma$]{
\includegraphics[width=8cm]{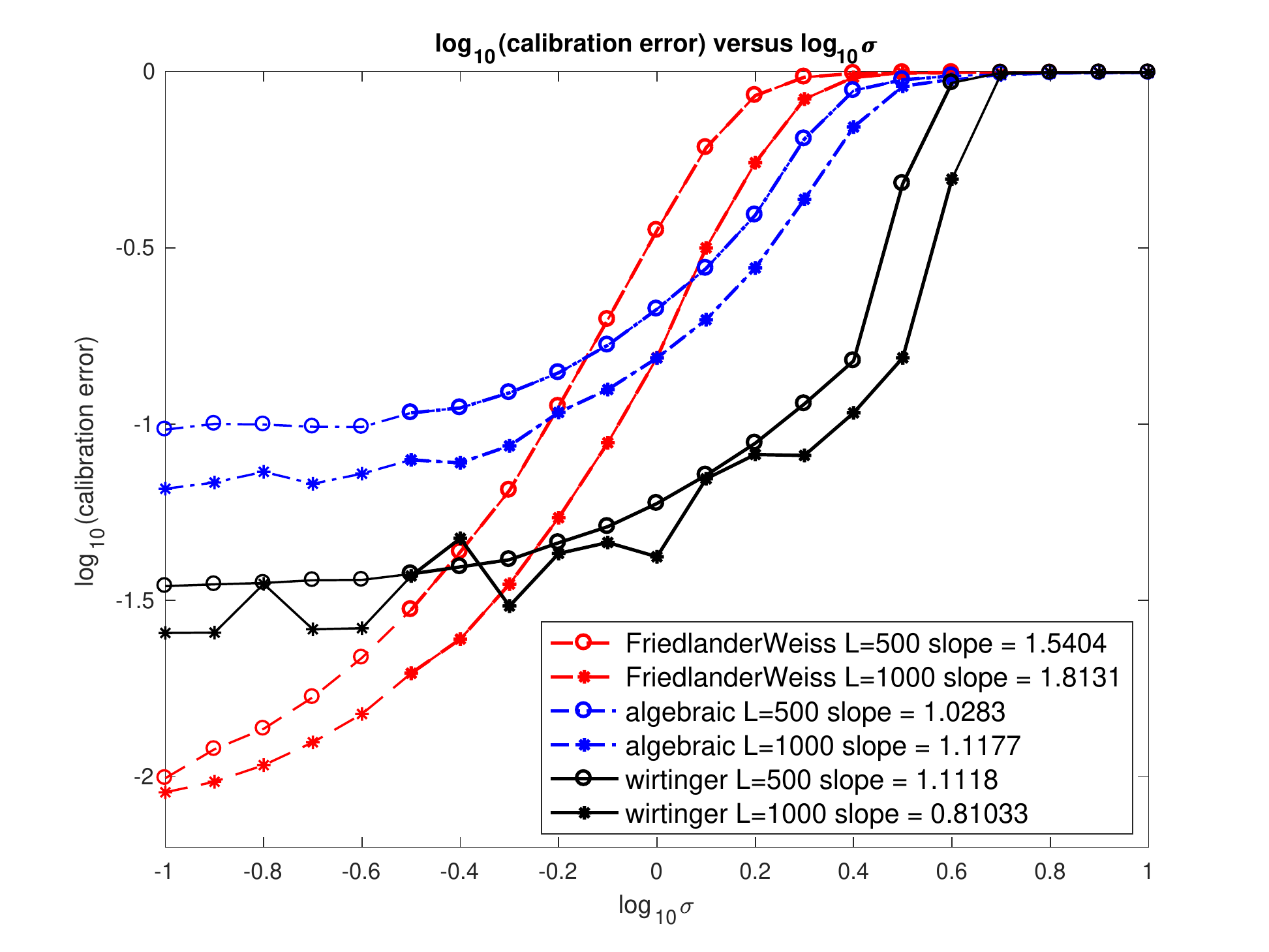}
}
\hspace{-1cm}
\subfigure[Frequency support success probability versus $\sigma$]{
\includegraphics[width=8cm]{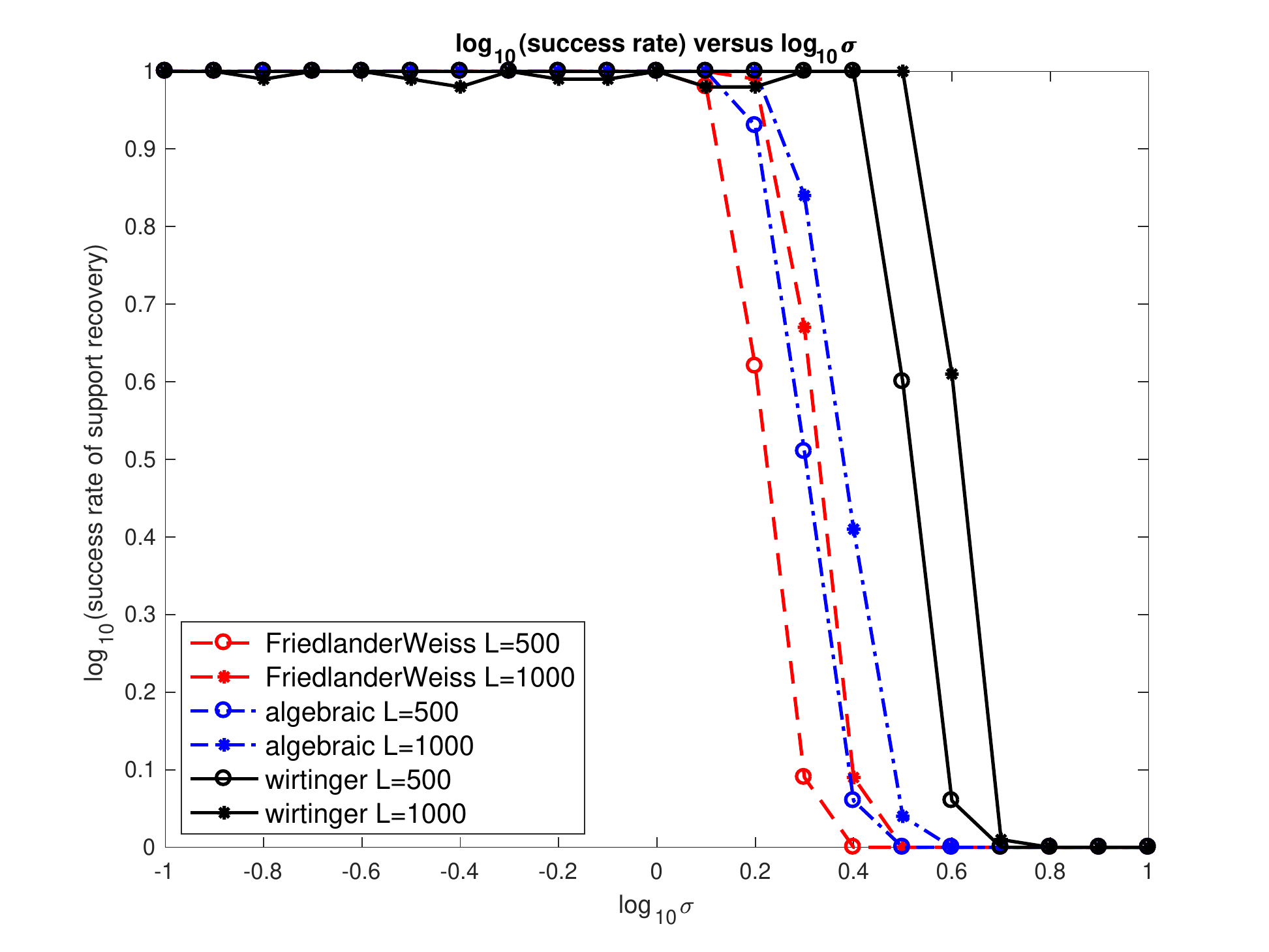}
}
\caption{Reconstruction errors of the calibration parameters and the success probability of support recovery versus $\sigma$ in $\log_{10}$ scale. We choose $20$ frequencies separated by $2/N$, $\DR_\gamma = \DR_g =2$, $L = 2000, 5000$, and $\sigma$ varies from $0.1$ to $10$. 
}
\label{FigErrVersusNoise}
\end{figure}

\section{Conclusion and future research}
\label{sec:conclusion}

This paper studies sensor calibration in spectral estimation with multiple snapshots. We assume the true frequencies are located on a continuous domain and each sensor has an unknown calibration parameter. Uniqueness of the calibration parameters and frequencies, up to a trivial ambiguity, is proved with infinite snapshots of noiseless measurements, based on the algebraic methods in \cite{PKailath,WRS93}. 
A sensitivity analysis of the partial algebraic method \cite{WRS93} with respect to the number of snapshots and noise is presented. 
While only partial measurements are exploited in the algebraic method, we propose an optimization approach to make full use of the measurements. Superior performance of our optimization approach is demonstrated through numerical comparisons with the partial algebraic method \cite{WRS93} and the alternating algorithm \cite{FriedlanderWeiss}.

Several interesting questions are left for future investigations. First, uniqueness in the current paper holds with infinite snapshots of noiseless measurements. It is interesting to study uniqueness with a minimal number of snapshots. Second,  global convergence of the Wirtinger gradient descent in our optimization approach is not proved in this paper, even though we have observed its superior numerical performance. 
The recent work in \cite{LLSW} guarantees global convergence of a non-convex optimization for the sensor calibration problem modeled by \eqref{modelling} where measurements are bilinear. In our problem, the covariance matrix is quadratic in $g$ and linear in $f$, which makes the local regularity condition \cite{CandesWirtinger,LLSW} harder to prove. 
%

\section*{Acknowledgement}
Wenjing Liao is supported by NSF-DMS-1818751 and a startup fund from Georgia Institute of Technology. Sui Tang is supported by the AMS Simons travel grant. Wenjing Liao and Sui Tang would like to thank Shuyang Ling for helpful discussions.

\appendix

\section{Sensitivity of the partial algebraic method (Proof of Theorem \ref{thmagstability})}
\label{appagstability}


The proof of Theorem \ref{thmagstability} relies on the following matrix Bernstein inequalities.
\begin{proposition}[{\cite[Theorem 7.3.1]{Tropp}}]
\label{propbern}
Consider a finite sequence $\{X_k\}$ of random Hermitian matrices that satisfy
$$\EE X_k = 0 \ \text{ and } \ \|X_k\| \le R.$$
Define the random matrix $Y = \sum_{k} X_k$. Suppose $\EE Y^2  \preceq V$ for some positive semidefinite matrix $V$ and let the intrinsic dimension of $V$ be ${\rm intdim}(V) := \trace(V)/\|V\|$. Then for any $t \ge \|V\|^{1/2}+R/3,$
\begin{align*}
\PP \left\{ \|Y\| \ge t\right\}& \le 4\cdot {\rm intdim}(V) \cdot \exp\left( 
\frac{-t^2/2}{\|V\|+Rt/3}
\right)
\\
\EE \|Y\| &\le \sqrt{2\|V\|\log(4\cdot {\rm intdim}(V))}+ \frac 1 3 R  \log(4\cdot {\rm intdim}(V)).
\end{align*}

\end{proposition}




\begin{proof}[Proof of Theorem \ref{thmagstability}]
In this proof, we assume $\|g\|$ is known, and $\beta_0 = \beta_{N-1} =0$ to remove trivial ambiguities; otherwise, our estimate gives an upper bound on $\min_{c_0>0,c_1,c_2\in \RR}\|c_0 \hg -e^{ic_1} \diag(\{e^{inc_2}\}_{n=0}^{N-1})g\|_\infty$.

In the case of finite snapshots, the sampled covariance matrix $\wtR^y_e$ deviates from $R^y_e$ by:
\begin{align*}
\|R^y_e - \widetilde R^y_e\|
& \le  \| G A (R^x - \widetilde R^x) A^* G^* + GA (R^{xe}-\wtR^{xe}) + (R^{ex}-\wtR^{ex}) A^* G^* + R^e - \wtR^e\| 
\\
&\le
\smaxsq(G) \smaxsq(A) \|R^x - \widetilde R^x\| + \smax(G)\smax(A) \|R^{xe}-\wtR^{xe}\| \\
& \quad + \smax(G)\smax(A) \|R^{ex}-\wtR^{ex}\|  +  \|R^e - \wtR^e\|,
\end{align*}
where we use the same notations as \eqref{eqcovRye}. 
We will estimate $\|R^x - \wtR^x\|$ using the matrix Bernstein inequality in Proposition \ref{propbern}. Let $\wtR_{x,t} =  \frac{1}{L}\left(x(t)x^*(t) - R^{x}\right) $ which satisfies $\|\wtR_{x,t} \| \le \frac{1}{L} \left(\max_{t \in \Gamma} \|x(t)\|^2+\gammamax^2 \right)$ for $t \in \Gamma$. Then $ \wtR^x-R^x = \sum_{t \in \Gamma} \wtR_{x,t}$. We observe that 
\begin{align*}
\EE ( \wtR^x-R^x  )^2 
&=  \sum_{t \in \Gamma} \EE \wtR_{x,t}^2 =\frac{1}{L^2} \sum_{t \in \Gamma}  \EE \left(x(t) x^*(t) - R^x\right)\left(x(t) x^*(t) - R^x\right)
\\ &
\preceq\frac{1}{L} \left(
\max_t \|x(t)\|^2 R^x -(R^x)^2 \right) \preceq \frac{\max_t \|x(t)\|^2}{L}R^x, 
\end{align*}
where $R^x$ has the intrinsic dimension ${\rm intdim}(R^x) \le s$. Applying Proposition \ref{propbern}, we obtain that for any 
$ \eta \ge \frac{\gammamax \max_t \|x(t)\|}{\sqrt{L}}+\frac{\max_{t\in \Gamma}{ \|x(t)\|^2}+\gammamax^2}{3L}$, we have
\begin{align}
\PP\{ \|R^x - \wtR^x \| \ge \eta \}& \le 4 s\cdot \exp\left(\frac{-\eta^2/2}{\frac{\gammamax^2\max_t \|x(t)\|^2}{L} + \frac{\gammamax^2+\max_{t\in \Gamma} \|x(t)\|^2}{3L}\eta}\right) \label{pthmagstability1},
\\
\EE\|R^x - \wtR^x \| &
\le \frac{ \gammamax \max_{t } \|x(t)\| \sqrt{2\log 4s}}{\sqrt{L}} + \frac{\gammamax^2 + \max_{t \in \Gamma} \|x(t)\|^2}{3L}\log 4s.
\label{pthmagstability10}
\end{align}
Similarly, for all $\eta >0$, we have 

\begin{align}
\PP \{\|R^{xe} - \wtR^{xe}\| \ge \eta \}
& \le (N+s) \cdot 
\exp\left( \frac{-\eta^2/2}{\frac{N\sigma^2\gammamax^2}{L}+ \frac{ \max_{t\in \Gamma} \|x(t)\|\|e(t)\|}{3L}\eta}\right),
\label{pthmagstability2}
\\
\PP\{\|R^e -\wtR^e \| \ge \eta \} 
&\le 2 N \cdot \exp\left( 
\frac{-\eta^2/2}{\frac{\sigma^2 \max_t \|e(t)\|^2}{L}+\frac{\sigma^2+\max_{t \in \Gamma} \|e(t)\|^2}{3L}\eta} 
\right),
\label{pthmagstability20}
\end{align}
and
\begin{align}
 \EE\|R^{xe}-\wtR^{xe}\| 
& \lesssim \frac{\sigma\gammamax\sqrt{2N\log (N+s)}}{\sqrt{L}} + \frac{\max_{t \in \Gamma}\|x(t)\|\|e(t)\| }{3L}\log (N+s),
\label{pthmagstability3}
\\
\EE \|R^e - \wtR^e\|
& \lesssim \frac{\sigma \max_{t } \|e(t)\|\sqrt{2\log 2N}}{\sqrt{L}}
+ \frac{\sigma^2+\max_{t\in \Gamma} \|e(t)\|^2}{3L}\log 2N, 
\label{pthmagstability30}
\end{align}
where we apply the matrix Bernstein inequality for the non-Hermitian case {\cite[Theorem 1.6.2]{Tropp}} to estimate $\|R^{xe} - \wtR^{xe}\|$. 
Our estimator of $R^y$ is $\hR^y = \wtR^y_e - \hsigma^2 I_N$, which has the following error 
\begin{align}
 \|R^y - \hR^y\| 
& =  \|(R^y_e - \sigma^2 I_N) - (\wtR^y_e-\hsigma^2 I_N)\|
\le  \|R^y_e -\wtR^y_e\|+ |\sigma^2 - \hsigma^2| \le 2 \|R^y_e -\wtR^y_e\|,
\label{pthmagstability31}
\end{align} where the last inequality follows from the  Weyl's inequality \cite{Weyl}. Combining \eqref{pthmagstability10}, \eqref{pthmagstability3}, \eqref{pthmagstability30} and \eqref{pthmagstability31} gives rise to $\EE\|R^y - \hR^y\| \le \Delta R^y$ with $\Delta R^y$ defined in \eqref{eqDeltaRy}.

Define the event 
$$\mathcal{E} :=\left\{
\max(\alphamax^2 \sigma^2_{\max}(A)\|R^x -\wtR^x\|, 
\alphamax\smax(A)\|R^{xe} -\wtR^{xe}\|,\|R^e -\wtR^e\|)\le \frac{\alphamin^2|f_1|}{16}\right\}$$
under which we have 
\begin{align*}
\|R^y -\hR^y \| \le \frac{1}{2}\alphamin^2 |f_1| \le \frac{1}{2}\alphamin^2 f_0.
\end{align*}
This implies
\begin{align*}
\trace(\hR^y) \ge \trace{R^y} - N\|R^y- \hR^y\|
= \|g\|^2 f_0 - \frac{1}{2}N \alphamin^2 f_0
\ge \frac{1}{2}\|g\|^2 f_0
= \frac 1 2 \trace(R^y),
\end{align*}
and $\trace(\hR^y) \le 3/2\trace(R^y)$.
We first perform all estimates under the event $\mathcal{E}$ and consider $\calE^c$ later.

\subsubsection*{Condition on $\calE$}
In the partial algebraic method, if $\|g\|$ is known, then the calibration amplitudes $\alpha$ can be recovered without any scaling ambiguity. The exact and recovered calibration phases are:
\beq
\alpha^2_n = \frac{R^y_{n,n}}{\trace(R^y) }\|g\|^2 \qquad 
\halpha^2_n = \frac{\hR^y_{n,n}}{\trace(\hR^y) }\|g\|^2.
\label{pthmagstability6}
\eeq
Hence
\begin{align*}
|\alpha_n^2 - \halpha_n^2|
& \le 
\left| \frac{R^y_{n,n}}{\trace(R^y) } - \frac{\hR^y_{n,n}}{\trace(\hR^y) } \right| \|g\|^2
=\left| \frac{R^y_{n,n}\trace(\hR^y) -\hR^y_{n,n}\trace(R^y)}{\trace(R^y)\trace(\hR^y) } \right| \|g\|^2
\\
& \le 
\frac{R^y_{n,n}|\trace(\hR^y) - \trace(R^y)| 
+
\trace(R^y) |\hR^y_{n,n}-R^y_{n,n}|
}{\trace(R^y)\trace(\hR^y) }  \|g\|^2
\\
& \le \frac{|g_n|^2 f_0 N  + \|g\|^2 f_0}{\|g\|^2 f_0 \|g\|^2 f_0 /2}\|g\|^2 \|R^y - \hR^y\|
=  2\frac{N |g_n|^2 +\|g\|^2}{\|g\|^2 f_0} \|R^y - \hR^y\|.
\end{align*}
On the other hand, 
$$\alpha_n^2 = \frac{|g_n|^2 f_0}{\trace{R^y}}\|g\|^2                           \qquad 
\halpha_n^2\ge \frac{R^y_{n,n} - \|R^y-\hR^y\|}{\frac 3 2 \trace{R^y} } \|g\|^2
\ge
\frac{|g_n|^2f_0 - \frac 1 2 |g_n|^2 f_0}{\frac 3 2 \trace{R^y} }  \|g\|^2
 \ge \frac{\alpha_n^2}{3},$$
in the event $\calE$ and then
\begin{align*}
|\alpha_n - \halpha_n| 
& = 
\frac{|\alpha_n^2 - \halpha_n^2|}{\alpha_n + \halpha_n} 
\le \frac{1}{4\alpha_n/3}
\cdot
2 \frac{\|g\|^2+N |g_n|^2}{\|g\|^2 f_0 } \cdot \|R^y - \hR^y\|,
\\
\|\alpha-\halpha\|_\infty 
& 
\le 
\frac{3}{2\alphamin}
\cdot
\frac{\|g\|^2+N \alphamax^2}{\|g\|^2 f_0 } \cdot  \|R^y - \hR^y\|.
\end{align*}

Next we estimate $\|\beta-\hbeta\|_\infty.$ To remove trivial ambiguities, we assume the exact calibration phases $\beta$ satisfy
\begin{align*}
\Phi \beta = b \quad \text{where } b_0=b_{N-1}=0, \ b_{n}= \angle \frac{R^y_{n+1,n}}{R^y_{n,n-1}}, n = 1,\ldots,N-2.
\end{align*}
Our recovered calibration phases $\hbeta$ are:
\begin{align*}
\Phi \hbeta = \hb \quad \text{where } \hb_0=\hb_{N-1}=0, \ \hb_{n}= \angle \frac{\hR^y_{n+1,n}}{\hR^y_{n,n-1}}, n = 1,\ldots,N-2.
\end{align*}
Recall that $R^y_{n,n-1} = \alpha_n \alpha_{n-1} e^{i (\beta_n - \beta_{n-1})}f_1$, so $\alphamin^2 |f_1| \le |R^y_{n,n-1}| \le \alphamax^2 |f_1|$. In the event $\calE$, we have $\alphamin^2 |f_1|/2 \le |\hR^y_{n,n-1}| \le 3\alphamax^2 |f_1|/2$, and 
\begin{align*}
\left| \frac{R^y_{n+1,n}}{R^y_{n,n-1}}  - \frac{\hR^y_{n+1,n}}{\hR^y_{n,n-1}}\right|
& = \frac{| R^y_{n+1,n}  \hR^y_{n,n-1} -R^y_{n,n-1}  \hR^y_{n+1,n} |}{|R^y_{n,n-1} \hR^y_{n,n-1}|} 
\le 4 \frac{\alphamax^2}{\alphamin^4 {|f_1|}} \|R^y - \hR^y\|.
\end{align*}
For any $z,\widehat z \in \CC$, by a simple geometric argument, we have
\begin{align}
 |(\angle z - \angle\widehat z) \mod 2 \pi | \le \frac{4|z-\widehat z|}{\min(|z|,|\widehat z|)}\end{align} whenever $ {|z-\widehat z|} \le  \min(|z|,|\widehat z|)$. Whenever $\|R^y -\hR^y\| \le \frac{ \alphamin^6 |f_1|}{12\alphamax^4}$ (This is guaranteed for sufficiently large $L$), 
 $$\left| \frac{R^y_{n+1,n}}{R^y_{n,n-1}}  - \frac{\hR^y_{n+1,n}}{\hR^y_{n,n-1}}\right|\leq \frac{\alphamin^2}{3\alphamax^2} \leq \min(\left| \frac{R^y_{n+1,n}}{R^y_{n,n-1}}\right|,  \left| \frac{\hR^y_{n+1,n}}{\hR^y_{n,n-1}}\right|). $$
 Hence
\begin{align*}
\|b-\hb\|_\infty
&= \max_n  \left| \left(\angle \frac{R^y_{n+1,n}}{R^y_{n,n-1}}  - \angle \frac{\hR^y_{n+1,n}}{\hR^y_{n,n-1}} \right) \mod 2\pi \right|  
\\&\le 
\max_n \frac{4\left| \frac{R^y_{n+1,n}}{R^y_{n,n-1}}  - \frac{\hR^y_{n+1,n}}{\hR^y_{n,n-1}}\right|}{\min\left(\left| \frac{R^y_{n+1,n}}{R^y_{n,n-1}}\right|, \left|\frac{\hR^y_{n+1,n}}{\hR^y_{n,n-1}}\right|\right)}
\le
\frac{48\alphamax^4 }{\alphamin^4}  \cdot
\frac{\|R^y -\hR^y\|}{\alphamin^2 {|f_1|}}.
\end{align*}
The infinity norm of the matrix $\Phi^{-1}$ is upper bounded by (see \cite[Chapter 2]{Leveque}):
$$\|\Phi^{-1}\|_\infty= \max_{j} \sum_{i=0}^{N-1} |\Phi^{-1}_{i,j}| \le 3 N^2.$$
Therefore
\begin{align*}
\|\beta-\hbeta\|_\infty \le \|\Phi^{-1}\|_\infty \|b-\hb\|_\infty
\le 
144 N^2 \frac{\alphamax^4 }{\alphamin^4}  \cdot
\frac{\|R^y-\hR^y\|}{\alphamin^2 {|f_1|}}.
\end{align*}
Combining the estimates of $\|\alpha-\halpha\|_\infty$ and $\|\beta-\hbeta\|_\infty$ gives rise to
\begin{align*}
\|g-\hg\|_\infty 
= \max_n |g_n - \hg_n| 
\le \max_n |\alpha_n| |e^{i\beta_n}-e^{i \hbeta_n}| +|\alpha_n -\halpha_n|
\le \|\alpha-\halpha\|_\infty + \alphamax \|\beta-\hbeta\|_\infty.
\end{align*}

As for the input matrix for the MUSIC algorithm, we have
$$F = \diag(g)^{-1} R^y \diag(\barg)^{-1}
\qquad 
\hF = \diag(\hg)^{-1} \hR^y \diag(\bar{\hg})^{-1}.$$
Then
\begin{align*}
\|F-\hF\|
&\le \frac{1}{\widehat\alpha_{\min}^2} \|R^y -\hR^y\|
+ \frac{\|R^y\|}{\alphamin} \max_n \left| \frac{1}{g_n} -\frac{1}{\hg_n}\right|
+ \frac{\|R^y\|}{\widehat\alpha_{\min}} \max_n \left| \frac{1}{g_n} -\frac{1}{\hg_n}\right| \\
& 
\le 9\frac{\|R^y -\hR^y\|}{\alphamin^2} + 12 \frac{\|R^y\|}{\alphamin^3}\|g-\hg\|_\infty 
\le 9 \frac{\|R^y -\hR^y\|}{\alphamin^2} + 12 \frac{\alphamax^2\gammamax^2  \sigma_{\max}^2(A)}{\alphamin^3}\|g-\hg\|_\infty.
\end{align*}
When the input of MUSIC is $\widehat F$, Proposition \ref{propmusic2} provides an estimate on the perturbation of the noise-space correlation function: 
$$|\widehat \calR(\om) - \calR(\om)| \le \frac{2 }{\gammamin^2\sigma^2_{\min}(A)}\cdot
\|F-\hF\| $$
 as long as $2\|F-\hF\| <  \gammamin^2\sigma^2_{\min}(A)$.

Conditioning on the event $\calE$, we have 
\begin{align*}
\EE(\|g-\hg\|_\infty |\calE) 
& \le
\EE( \|\alpha-\halpha\|_\infty |\calE) 
+ \alphamax
\EE( \|\beta-\hbeta\|_\infty |\calE)
\\&
 \le
\frac{3(\|g\|^2+N \alphamax^2)}{2\alphamin\|g\|^2 f_0 } \Delta R^y
+
144 N^2 \frac{\alphamax^5 }{\alphamin^6 {|f_1|}}  \Delta R^y
\end{align*}
and
\begin{align*}
\EE(\|F-\hF\| | \calE) 
&\le
9\frac{\Delta R^y}{\alpha^2_{\min}}
+  \frac{12 \alphamax^2\gammamax^2  \sigma^2_{\max}(A)}{\alphamin^3}
\left(
\frac{3(\|g\|^2+N \alphamax^2)}{2\alphamin\|g\|^2 f_0 } +
144 N^2 \frac{\alphamax^5 }{\alphamin^6 {|f_1|}}  
\right) \Delta R^y.
\end{align*}

\subsubsection*{Condition on $\calE^c$}
Finally we consider the event $\calE^c$ which occurs with small probability when $L$ is sufficiently large:
\begin{align*}
\PP\{\calE^c\} 
&\le 
\PP \left\{ \|R^x - \wtR^x\| \ge \frac{\alphamin^2 |f_1|}{16 \alphamax^2 \sigma^2_{\max}(A)} \right\} 
+ \PP 
\left\{\| R^{xe} -\wtR^{xe}\| \ge \frac{\alphamin^2 |f_1|}{16\alphamax \smax(A)}
\right\}
\\
&
+ \PP\left\{
\|R^e - \wtR^e\| \ge \frac{\alphamin^2 |f_1|}{16}
\right\}
\\
& \le 
4N e^{-LC(\alphamax,\alphamin,\gammamax,\smax(A),|f_1|,\sigma,\max_t \|x(t)\|,\max_t \|e(t)\|)}
\end{align*}
for some positive constant $C(\alphamax,\alphamin,\gammamax,\smax(A),|f_1|, \sigma,\max_t \|x(t)\|,\max_t \|e(t)\|).$ In any case, $\|g-\hg\|_\infty \le \|g\|_\infty+\|\hg\|_\infty \le \|g\|_\infty + \|g\| \le 2\|g\|,$ where $\|\hg\|_\infty \le \|g\|$ due to \eqref{pthmagstability6}. Therefore,
\begin{align}
\EE\|g-\hg\|_\infty 
& \le \EE(\|g-\hg\|_\infty | \calE)  \PP \{\calE\} + 2 \|g\| \PP\{\calE^c\}
\nn
\\
&\le \frac{3(\|g\|^2+N \alphamax^2)}{2\alphamin\|g\|^2 f_0 } \Delta R^y
+
144 N^2 \frac{\alphamax^5 }{\alphamin^6 |f_1|^2}  \Delta R^y
+ 8N \|g\|_2  e^{-LC}.
\label{pthmagstability5}
\end{align}
Since the the first two terms in \eqref{pthmagstability5} is $O(1/L)$ and the last term is $O(e^{-CL})$, we can guarantee \eqref{eqdeltag}
when $L$ is sufficiently large. A similar estimate holds for $\|F-\hF\|$.

\end{proof}

\section{Proof of Lemma \ref{lemman0}}
\label{prooflemman0}
According to \eqref{eqRymn}, we have
$$|g_n|^2 = \frac{R^y_{n,n}}{f_0}, \ n=0,\ldots,N-1,$$
and for any $k = 1,\ldots,N-1$
$$|f_k|^2 = \frac{|R^y_{n+k,n}|^2}{|g_{n+k}|^2 |g_n|^2} = \frac{f_0^2|R^y_{n+k,n}|^2}{R^y_{n+k,n+k} R^y_{n,n}} \ \text{ for any } 0 \le n \le N-k-1. $$
Therefore
$$\|g\|^2 = \frac{\sum_{n=0}^{N-1}R^y_{n,n}}{f_0}$$
and 
\begin{align*}
\|f\| &= \sqrt{\sum_{k=0}^{N-1} |f_k|^2}  = \sqrt{f_0^2 + \sum_{k=1}^{ N-1} \frac{1}{N-k} \sum_{n=0}^{N-k-1} \frac{f_0^2|R^y_{n+k,n}|^2}{R^y_{n+k,n+k} R^y_{n,n}}}
\\
&=f_0  \sqrt{1 +  \frac{1}{N-k}\sum_{k=1}^{ N-1} \sum_{n=0}^{N-k-1} \frac{|R^y_{n+k,n}|^2}{R^y_{n+k,n+k} R^y_{n,n}}}
\end{align*}
which gives rise to Lemma \ref{lemman0}.

\section{Proof of Theorem \ref{thmwirtinger}}
\label{appthmwirtinger}

We first show that $\nabla \tcalL$, restricted within $\calN_{\hn_0}$, is a Lipchitz function. Notice that $g$ are the exact calibration parameters and $f$ is defined in \eqref{eqf}.
\begin{lemma}
\label{lemmalip}
For any $z := (\bg;  \bff)$ and $\Delta \bz:= (\Delta \bg; \Delta \bff)$ such that $\bz,\bz+\Delta \bz \in \calN_{\hn_0}$, $\nabla \tcalL$ is Lipchitz such that
$$\|\nabla \tcalL (\bz+\Delta \bz) - \nabla \tcalL(\bz)\| \le C_{\rm Lip} \|\Delta \bz\|$$
with
$$C_{\rm Lip} \le 146\hn_0 \max(\sqrt{\hn_0},\sqrt[4]{\hn_0})+8\hn_0 + 16\max(\sqrt{\hn_0},\sqrt[4]{\hn_0})\|R^y -\hR^y\|_F + \frac{8\rho}{\min(\hn_0,\sqrt{\hn_0})}
$$
where 
$\rho \ge \frac{3\hn_0 +\|R^y -\hR^y\|_F}{(\sqrt 2 -1)^2}.$
\end{lemma}

\begin{proof}[Proof of Lemma \ref{lemmalip}]
The Wirtinger gradient of $\tcalL$ is 
\beq
\label{lemmalippeq1}
\nabla \tcalL = (\nabla_{\bg} \tcalL  ; \   \nabla_{\bff} \tcalL ) = (\nabla_{\bg} \calL + \nabla_{\bg} \calG  ; \ \nabla_{\bff} \calL + \nabla_{\bff} \calG).
\eeq
\begin{description}

\item[Part 1:] We estimate $\|\nabla_{\bg} \calL(\bz+\mathbf{w})-\nabla_{\bg} \calL(\bz)\|$. Recall that $R^y = \EE y(t)y(t)^*$, and
$$\nabla_{\bg} \calL(z) = 2 \diag \Big[
\overline{
\calT(\bff)^*  \diag(\bar \bg) 
\Big( \diag(\bg)\calT(\bff)\diag(\bar \bg) -  \diag(g)\calT(f)\diag(\bar{g}) + R^y - \hR^y
\Big)}
\Big].$$
Notice that for any $\bff,\bg,\bg_1,\bg_2,\bh \in \CC^N$ and $X \in \CC^{N \times N}$, we have 
\begin{align*}
\|\diag[\calT(\bff)^* \diag(\bg_1) \calT(\bh) \diag(\bg_2)]\| 
&\le \|\bff\|\|\bg_1\|\|\bh\|\|\bg_2\|
\\
\left\|\diag\left[ \calT(\bff)^* \diag(\bg) X\right]\right\| 
& \le  \sqrt 2 \|\bff\| \|\bg\| \|X\|_F.
\end{align*}
For any $\bz,\bz+\Delta \bz \in \calN_{\hn_0}$, we have
\begin{align}
&\|\nabla_{\bg} \calL(\bz+\Delta \bz)-\nabla_{\bg} \calL(\bz)\| \nn
\\
&\le
2 \Big\| \diag\Big[\calT(\bff+\Delta \bff)^* \diag(\bar\bg+\overline{\Delta \bg}) \diag(\bg+\Delta \bg) \calT(\bff+\Delta \bff) \diag(\bar\bg+\overline{\Delta \bg})   \nn
\\
& \quad
-
\calT(\bff)^* \diag(\bar\bg) \diag(\bg) \calT(\bff) \diag(\bar\bg)  \Big]
\Big\| \nn
\\
&\quad 
+ 2\left\| \diag\Big[ \calT(\bff+\Delta \bff)^* \diag(\bar\bg +\overline{\Delta \bg}) \diag(g)\calT(f)\diag(\bar{g}) -  \calT(\bff)^* \diag(\bar\bg ) \diag(g)\calT(f)\diag(\bar{g})  \Big] \right\| \nn
\\
& \quad
+
2\left\| \diag\Big[ \calT(\bff+\Delta \bff)^* \diag(\bar\bg +\overline{\Delta \bg}) (R^y -\hR^y)-\calT(\bff)^* \diag(\bar\bg ) (R^y -\hR^y)  \Big] \right\| \nn
\\
& \le
2\Big( \|\Delta \bff\| \|\bg+ \Delta \bg\|^3\|\bff+ \Delta \bff\|+ \|\bff\|\|\Delta \bg\|\|\bg+ \Delta \bg\|^2\|\bff+\Delta \bff\|+\|\bff\|\|\bg\|\|\Delta \bg\|\|\bg+\Delta \bg\|\|\bff+\Delta \bff\| \nn
\\
& \quad +
\|\bff\|\|\bg\|^2\|\Delta \bff\|\|\bg+\Delta \bg\|+\|\bff\|^2\|\bg\|^2 \|\Delta \bg\|
+
\|\Delta \bff\|\|\bg+\Delta \bg\|\|g\|^2\|f\| + \|\bff\|\|\Delta \bg\|\|g\|^2\|f\| \Big)
\nn
\\
& 
\quad+ 2\sqrt 2\left(
\|\Delta \bff\| \bg+\Delta \bg\| \|R^y -\hR^y\|_F +\|\bff\| \|\Delta \bg\| \|R^y -\hR^y\|_F \right) \nn
\\
&
\le 64\hn_0\sqrt{\hn_0} \|\Delta \bg\| + 24\sqrt{2}\hn_0 \hn_0^{\frac 1 4} \|\Delta \bff\| + 4\hn_0^{\frac 1 4} \|R^y -\hR^y\|_F\|\Delta \bff\| + 4\sqrt{2\hn_0}\|\Delta \bg\| \|R^y -\hR^y\|_F.
\label{lemmalippeq2}
\end{align}

\item[Part 2:] We estimate $\|\nabla_{\bff} \calL(\bz+\Delta \bz)-\nabla_{\bff} \calL(\bz)\|$. 
Notice that for any $\bg_1,\bg_2,\bff \in \CC^N$ and $X \in \CC^{N \times N}$, we have 
\begin{align*}
\|\calT^a[\diag(\bg_1)\calT(\bff)\diag(\bg_2)]\| & \le 2\|\bg_1\|\|\bff\|\|\bg_2\|,
\\
\|\calT^a[\diag(\bg_1)X \diag(\bg_2)]\| &\le 2 \|\bg_1\| \|\bg_2\| \|X\|_F.
\end{align*}
By using triangle inequalities, we obtain
\begin{align*}
&\|\nabla_{\bff} \calL(\bz+\Delta \bz)-\nabla_{\bff} \calL(\bz)\| \\
& \le
\Big\| \calT^a\Big[ \diag(\overline{\bg+\Delta \bg})\diag(\bg+\Delta \bg)\calT(\bff+\Delta \bff)\diag(\overline{\bg+\Delta \bg})\diag(\bg+\Delta \bg)
\\
& \quad - \diag(\bar\bg)\diag(\bg)\calT(\bff)\diag(\bar\bg)\diag(\bg)\Big]
\Big\|
\\
& 
\quad + 
\Big\| \calT^a \Big[\diag(\overline{\bg+\Delta \bg})   \diag(g)\calT(f)\diag(\bar{g}) \diag(\bg+\Delta \bg) - \diag(\bar\bg)   \diag(g)\calT(f)\diag(\bar{g}) \diag(\bg) \Big] \Big\|
\\
&\quad + 
\Big\| \calT^a \Big[ \diag(\overline{\bg+\Delta \bg}) ( R^y - \hR^y ) \diag(\bg+\Delta \bg) 
- \diag(\bar\bg) ( R^y - \hR^y ) \diag(\bg) \Big] \Big\|
\\
& \le 
2\Big\| |\bg+\Delta \bg|^2 - |\bg|^2 \Big\| \cdot \|\bff+\Delta \bff\|  \cdot \Big\| |\bg+\Delta \bg|^2 \Big\|
+2 \Big\| |\bg|^2 \Big\| \cdot \|\Delta \bff\| \cdot \Big\|  |\bg+\Delta \bg|^2 \Big\|
\\
& \quad + 2  \Big\| |\bg|^2 \Big\| \cdot \| \bff\| \cdot \Big\| |\bg+\Delta \bg|^2 - |\bg|^2 \Big\|
 + 2 \|\Delta \bg\| \cdot \|g\|^2 \cdot \|f\| \cdot \|\bg+\Delta \bg\|
+ 2 \|\bg\| \cdot \|g\|^2 \cdot \|f\| \cdot \|\Delta \bg\|
\\
& \quad
+2 \|\Delta \bg\|  \cdot \|R^y -\hR^y\|_F \cdot \|\bg+\Delta \bg\| + 2 \|\bg\|  \cdot  \|R^y -\hR^y\|_F  \cdot\|\Delta \bg\|.
\end{align*}
Whenever $\bz,\bz+\Delta \bz \in \calN_{\hn_0}$, we have 
$$\|\Delta \bg\| \le 2 \sqrt{2}\hn_0^{\frac 1 4}, \ \Big\| |\bg+\Delta \bg|^2 - |\bg|^2 \Big\| \le 2\sqrt{2}\hn_0^{\frac 1 4} \|\Delta \bg\|$$
and therefore
\beq
\label{lemmalippeq3}
\|\nabla_{\bff} \calL(\bz+\Delta \bz)-\nabla_{\bff} \calL(\bz)\| 
\le 
48 \sqrt 2 \hn_0 \hn_0^{\frac 1 4} \|\Delta \bg\| + 8\hn_0 \|\Delta \bff\| + 4\sqrt{2} \hn_0^{\frac 1 4} \|R^y -\hR^y\|_F \|\Delta \bg\|.
\eeq

\item[Part 3:]  We estimate $\|\nabla_{\bff} \calG(\bz+\Delta \bz)-\nabla_{\bff} \calG(\bz)\|$ and $\|\nabla_{\bg} \calG(\bz+\Delta \bz)-\nabla_{\bg} \calG(\bz)\|$. Notice that $\calG_0'(z) = 2\max(z-1,0)$ and hence
$$|\calG_0'(z_1) - \calG_0'(z_2)| \le 2|z_1-z_2|, \quad \calG_0'(z) \le 2 |z|, \quad \forall z_1,z_2,z \in \RR.$$
For any $\bz,\bz+\Delta \bz \in \calN_{\hn_0}$, we have
\begin{align}
& \|\nabla_{\bff} \calG(\bz+\Delta \bz)-\nabla_{\bff} \calG(\bz)\|
= \frac{\rho}{2\hn_0} \left\| \calG_0'\left(\frac{\|\bff+\Delta \bff\|^2}{2\hn_0} \right)(\bff+\Delta \bff) - \calG_0'\left(\frac{\|\bff\|^2}{2\hn_0} \right) \bff
\right\| \nn
\\
&
\le \frac{\rho}{2\hn_0} \left| \calG_0'\left(\frac{\|\bff+\Delta \bff\|^2}{2\hn_0} \right) - \calG_0'\left(\frac{\|\bff\|^2}{2\hn_0} \right) \right| \|\bff+\Delta \bff\|
+ \frac{\rho}{2\hn_0} \calG_0'\left(\frac{\|\bff\|^2}{2\hn_0} \right)\|\Delta \bff\| \nn
\\
& \le
\frac{\rho}{2\hn_0} \frac{2(\|\bff+\Delta \bff\|+\|\bff\|)(\|\bff+\Delta \bff\|-\|\bff\|)}{2\hn_0} \|\bff+\Delta \bff\|
+
\frac{\rho}{2\hn_0} \cdot \frac{2\|\bff\|^2}{2\hn_0} \|\Delta \bff\|
 \le \frac{6\rho}{\hn_0} \|\Delta \bff\|.
 \label{lemmalippeq4}
\end{align}
and 
\beq
\label{lemmalippeq5} 
\|\nabla_{\bg} \calG(\bz+\Delta \bz)-\nabla_{\bg} \calG(\bz)\|
\le \frac{6\rho}{\sqrt{\hn_0}}\|\Delta \bg\|.
\eeq

\end{description}

Combining \eqref{lemmalippeq1}, \eqref{lemmalippeq2}, \eqref{lemmalippeq3}, \eqref{lemmalippeq4}, \eqref{lemmalippeq5} gives rise to
\begin{align*}
&\|\nabla \tcalL(\bz+\Delta \bz) - \nabla \tcalL(\bz) \|
\\
 &\le
\left(
64\hn_0\sqrt{\hn_0} + 48 \sqrt{2}\hn_0 \sqrt[4]{\hn_0}
+ 4 \sqrt{2 \hn_0} \|R^y - \hR^y\|_F + 4\sqrt 2 \sqrt[4]{\hn_0} \|R^y - \hR^y\|_F
+ \frac{6\rho}{\sqrt{\hn_0}}
\right) \|\Delta \bg\|
\\
&\quad + 
\left(
24\sqrt{2} \hn_0 \sqrt[4]{\hn_0} + 8 \hn_0+ 4 \sqrt[4]{\hn_0} \|R^y - \hR^y\|_F + \frac{6\rho}{\hn_0}
\right) \|\Delta \bff\|
\\
& 
\le \left(166\hn_0 \max(\sqrt{\hn_0},\sqrt[4]{\hn_0})+8\hn_0 + 16\max(\sqrt{\hn_0},\sqrt[4]{\hn_0})\|R^y -\hR^y\|_F + \frac{12\rho}{\min(\hn_0,\sqrt{\hn_0})}
 \right) \|\Delta \bz\|.
\end{align*}

\end{proof}

The proof of Theorem \ref{thmwirtinger} is given below.

\begin{proof}[Proof of Theorem \ref{thmwirtinger}]
This proof consists of two parts. In Part 1, we will prove that $(\bg^k,\bff^k) \in \calN_{\hn_0}$ for every $k$, so $\nabla \tcalL$ always satisfies the Lipchitz property in Lemma \ref{lemmalip}. In Part 2, we prove the convergence of the gradient descent algorithm.

\begin{description}
\item[Part 1:]
In the optimization approach, we assume $\hn_0 = \|g\|^2\|f\|$ is known, and start with an initial point $(\bg^0,\bff^0)$ satisfying $\|\bg^0\| \le \sqrt[4]{2 \hn_0}, \|\bff^0\| \le \sqrt{2 \hn_0}.$ Notice that for any $\bg,\bff,\bh \in \CC^N$, we have 
$$\|\diag(\bg)\calT(\bff)\diag(\bh)\|_F \le \|\bg\|\|\bff\|\|\bh\|,$$
and then
\begin{align}
\tcalL(\bg^0,\bff^0) 
&= \calL(\bg^0,\bff^0) = \|\diag(g^0)\calT(\bff^0)\diag(\overline{\bg^0})-\diag(g)\calT(f)\diag(\overline{g})+R^y-\hR^y\|_F
\nn
\\
& \le \|\bg^0\|^2\|\bff^0\| + \|g\|^2\|f\| + \|R^y - \hR^y\|_F
\le 3 \hn_0 +  \|R^y - \hR^y\|_F.
\label{thmwirtingerpeq1}
\end{align}
Our gradient descent algorithm guarantees $\tcalL(\bg^k,\bff^k) \le \tcalL(\bg^0,\bff^0), k=1,2,\ldots$ (see \eqref{thmwirtingerpeq2}). We will prove $(\bg^k,\bff^k) \in \calN_{\hn_0}$ by contradiction. Assume that $(\bg^k,\bff^k) \notin \calN_{\hn_0}$ for some $k$. Then
\begin{align*}
\tcalL(\bg^k,\bff^k) \ge \rho\left[\calG_0\left(\frac{\|\bff^k\|^2}{2\hn_0} \right) +\calG_0\left(\frac{\|\bg^k\|^2}{\sqrt{2\hn_0}} \right) \right] > \rho\calG_0(\sqrt 2) = \rho(\sqrt 2 -1)^2.
\end{align*}
By taking $\rho \ge \frac{3\hn_0 +\|R^y -\hR^y\|_F}{(\sqrt 2 -1)^2}$, we would have $\tcalL(\bg^k,\bff^k) > 3\hn_0 + \|R^y -\hR^y\|_F$ which contradicts \eqref{thmwirtingerpeq1}. We conclude that $(\bg^k,\bff^k) \in \calN_{\hn_0}$ at every iteration $k$.

\item[Part 2:] Let $\bz = (\bg,\bff)$ and $\Delta \bz = (\Delta \bg, \Delta \bff)$. Notice that $\tcalL$ is continuously differentiable and real-valued. Suppose $\bz,\bz+\Delta \bz \in \calN_{\hn_0}$. Then $\bz+t\Delta \bz\in\calN_{\hn_0}$ due to convexity of $\calN_{\hn_0}$.

It follows from Lemma 6.1 in \cite{LLSW} that, if $h(t) := \tcalL(\bz+t\Delta \bz)$, then
$$\frac{d h(t)}{dt} = (\Delta \bz)^T \frac{\partial \tcalL}{\partial \bz} (\bz+t\Delta \bz)
+ (\Delta \bar \bz)^T \frac{\partial \tcalL}{\partial \bar \bz}(\bz+t\Delta \bz) = 2 {\rm Re}\left((\Delta \bz)^T\overline{ \nabla_{\bz} \tcalL (\bz+t\Delta \bz)}\right).$$
By the Fundamental Theorem of Calculus, we have
\begin{align*}
 \tcalL(\bz+\Delta \bz) - \tcalL(\bz)
& = \int_0^1 \frac{d h(t)}{dt} dt 
= 2  \int_0^1 {\rm Re}\left((\Delta \bz)^T\overline{ \nabla_{\bz} \tcalL (\bz+t\Delta \bz)}\right) dt 
\\
& \le  2  {\rm Re}\left((\Delta \bz)^T\overline{ \nabla_{\bz} \tcalL (\bz)}\right)  
+ 2 \|\Delta \bz\| \int_0^1 \| \nabla_{\bz} \tcalL (\bz+t\Delta \bz) - \nabla_{\bz} \tcalL (\bz)\| dt 
\\
& \le 
2  {\rm Re}\left((\Delta \bz)^T\overline{ \nabla_{\bz} \tcalL (\bz)}\right) 
+ C_{\rm Lip}  \|\Delta \bz\|^2.
\end{align*}

At the $k$th iteration, we let $\bz=(\bg^k,\bff^k)$, and $\Delta \bz =- \eta^k \nabla_{\bz} \tcalL(\bg^k,\bff^k)$, and then
\beq
\tcalL(\bg^{k+1},\bff^{k+1}) \le \tcalL(\bg^k,\bff^k)  - (2-C_{\rm Lip}\eta^k)\eta^k \|\nabla_{\bz} \tcalL(\bg^k,\bff^k)\|^2.
\label{thmwirtingerpeq2}
\eeq
As long as $\|\nabla_{\bz} \tcalL(\bg^k,\bff^k)\|>0$ and $\eta^k < 2/C_{\rm Lip}$, we have 
$$\tcalL(\bg^{k+1},\bff^{k+1}) < \tcalL(\bg^k,\bff^k),$$
which implies $\|\nabla_{\bz} \tcalL(\bg^k,\bff^k)\| \rightarrow 0$
as $k \rightarrow \infty.$ This captures the proof that the Wirtinger gradient descent converges to a critical point. 
\end{description}

\end{proof}

\commentout{
\section{Wirtinger derivatives}
\label{appwd}
Notice that $\calL(\bg,\bff) = \langle \diag(\bg)\calT(\bff)\diag(\bar\bg) -\hR^y,\diag(\bg)\calT(\bff)\diag(\bar\bg) -\hR^y\rangle$, and then
\begin{align*}
\frac{\partial \calL}{\partial \bg_n} 
& =
2\left\langle \diag(\bg)\calT(\bff)\frac{\partial \diag(\bar \bg)}{\partial \bar \bg_n} , \diag(\bg)\calT(\bff)\diag(\bar \bg) - \hR^y\right\rangle 
\\
& =2\left[ \calT(\bff)^*\diag(\bar\bg)\left( \diag(\bg)\calT(\bff)\diag(\bar \bg) - \hR^y \right)\right]_{n,n}
\end{align*}
which gives rise to \eqref{eqdg}. By a similar calculation, we have
\begin{align*}
\frac{\partial \calL}{\partial \bff_n}
=&
\left \langle \diag(\bg)\frac{\partial \calT(\bff)}{\partial \bar \bff_n}\diag(\bar \bg) , \diag(\bg)\calT(\bff)\diag(\bar \bg) - \hR^y \right\rangle \\
& +
\left \langle  \diag(\bg)\calT(\bff)\diag(\bar \bg) - \hR^y , \diag(\bg)\frac{\partial \calT(\bff)}{\partial \bff_n}\diag(\bar \bg)\right\rangle
\end{align*}
which along with \eqref{Ha} give rise to \eqref{Def}.
}

\commentout{
\section{Hessian Matrix}
In this section, we compute the Hessian matrix of the objective function. We first compute the  Hessian matrix with respect to the variables $g$ only.
\begin{align*}
 \frac{ \partial \calL}{2\partial g_n \partial \bar{g}_n} 
 &= \sum_{i=0}^{n} |g_i|^2 |f_{n-i}|^2+\sum_{j=1}^{N-n-1} |g_{n+j}|^2|f_j|^2+f_0(g_n\bar g_n \bar f_0-\bar {R}_{n,n}^y ), \\
 \frac{ \partial \calL}{2\partial g_n \partial {g}_n} 
 &= \overline {g_n}^2 |f_0|^2\\
  \frac{ \partial \calL}{2\partial g_n \partial \bar{g}_{m}}
&= 
\begin{cases} 
f_{n-m} (\barg_n g_m \barf_{n-m} -\bar{R}^y_{n,m}) & \mbox{ if } m=0,\dots,n-1\\
\bar f_{m-n}(g_{m}\bar g_{n} f_{m-n}-{R}_{m,n}^y), &\mbox{ if }m=n+1, \cdots, N-1
\end{cases}
\\
 \frac{ \partial \calL}{2\partial g_n \partial {g}_m} &= \bar{g}_m |f_{m-n}|^2 \bar g_n,~~~~~~ m=n+1,\cdots, N-1. \\
 \frac{ \partial \calL}{2\partial g_n \partial \bar{f}_m} &= \begin{cases}  | g_{n-m}|^2 \bar g_n  f_m    &\mbox{ if } N-n \leq m \leq n \mbox{ or } m=0 \\ 
\bar g_{n+m}( g_{n+m} \bar g_n  f_m- {R}^y_{n+m,n})  & \mbox{ if } n+1 \leq m \leq N-n-1 \\
 | g_{n-m}|^2 \bar g_n  f_m   + \bar g_{n+m}( g_{n+m} \bar g_n  f_m- {R}^y_{n+m,n}) & \mbox{ if } 1\leq m \leq \min\{n, N-n-1\}\\
 0  & \mbox{ otherwise} \end{cases}\\
 \frac{ \partial \calL}{2\partial g_n \partial {f}_m} &= \begin{cases} \bar g_{n-m}(g_{n-m}\bar g_n \bar f_m-\bar{R}_{n,n-m}^y)     &\mbox{ if } N-n \leq m \leq n  \mbox{ or } m=0  \\ 
|g_{n+m}|^2\bar f_m \bar g_n  & \mbox{ if } n+1 \leq m \leq N-n-1 \\
\bar g_{n-m}(g_{n-m}\bar g_n \bar f_m-\bar{R}_{n,n-m}^y) + |g_{n+m}|^2\bar f_m \bar g_n   & \mbox{ if } 1\leq m \leq \min\{n, N-n-1\}\\
 0  & \mbox{otherwise} \end{cases}\\
 \text{Case I: } & n \le N-n+1  \\
  \frac{ \partial \calL}{2\partial g_n \partial {f}_m} &= \begin{cases}
  |g_{n+m}|^2\bar f_m \bar g_n  & \mbox{ if } n+1 \leq m \leq N-n-1 \\
\bar g_{n-m}(g_{n-m}\bar g_n \bar f_m-\bar{R}_{n,n-m}^y) + |g_{n+m}|^2\bar f_m \bar g_n   & \mbox{ if } 1\leq m \leq \min\{n, N-n-1\}\\
 0  & \mbox{otherwise} \end{cases}\\
  \frac{ \partial \calL}{2\partial g_n \partial \bar{f}_m} &= \begin{cases}   
\bar g_{n+m}( g_{n+m} \bar g_n  f_m- {R}^y_{n+m,n})  & \mbox{ if } n+1 \leq m \leq N-n-1 \\
 | g_{n-m}|^2 \bar g_n  f_m   + \bar g_{n+m}( g_{n+m} \bar g_n  f_m- {R}^y_{n+m,n}) & \mbox{ if } 1\leq m \leq \min\{n, N-n-1\}\\
 0  & \mbox{ otherwise} \end{cases}\\
 \text{Case II: } & n > N-n+1  \\
  \frac{ \partial \calL}{2\partial g_n \partial {f}_m} &= \begin{cases}
\bar g_{n-m}(g_{n-m}\bar g_n \bar f_m-\bar{R}_{n,n-m}^y)   & \mbox{ if } N-n \le  m  \le n   \\
\bar g_{n-m}(g_{n-m}\bar g_n \bar f_m-\bar{R}_{n,n-m}^y) + |g_{n+m}|^2\bar f_m \bar g_n   & \mbox{ if } 1\leq m \leq \min\{n, N-n-1\}\\
 0  & \mbox{otherwise} \end{cases}\\ 
  \frac{ \partial \calL}{2\partial g_n \partial \bar{f}_m} &= \begin{cases}   
| g_{n-m}|^2 \bar g_n  f_m & \mbox{ if } N-n \le m \leq n \\
 | g_{n-m}|^2 \bar g_n  f_m   + \bar g_{n+m}( g_{n+m} \bar g_n  f_m- {R}^y_{n+m,n}) & \mbox{ if } 1\leq m \leq \min\{n, N-n-1\}\\
 0  & \mbox{ otherwise} \end{cases}\\
 \frac{ \partial \calL}{2\partial  \bar g_n \partial {g}_{m}}
&=\bar f_{m-n}(g_{m}\bar g_{n} f_{m-n}-{R}_{m,n}^y), ~~~~~m=n+1, \cdots, N-1.\\
 \frac{ \partial \calL}{2\partial  \bar g_n \partial {f}_{m}}
&= \frac{ \partial \calL}{2\partial   g_n \partial \bar{f}_{m}} \\
 \frac{ \partial \calL}{2\partial  \bar g_n \partial \bar{f}_{m}}
&= \frac{ \partial \calL}{2\partial   g_n \partial {f}_{m}} \\
\frac{ \partial \calL}{2\partial f_n \partial \bar{f}_i}
&=0, ~~~~~~~~i\neq n\\
 \frac{ \partial \calL}{2\partial f_n \partial \bar{f}_n} 
 &= \sum_{i=0}^{N-n-1} |g_i|^2 |g_{n+i}|^2 \\
  \frac{ \partial \calL}{2\partial f_n \partial {f}_m} 
 &= 0  \\
\end{align*}
}

\commentout{
\section{Local smoothness}

In this section, we want to prove the gradient of the objective function satisfies the Lipschitz condition on the bounded domain.
\begin{align*}
\frac{ \partial \calL}{\partial g } (f+\Delta f, g+\Delta g)-\frac{ \partial \calL}{\partial g} (f, g)&=\mathcal{H}^a(X_1+X_2+X_3),
\text{ where } \\
X_1&=\diag(| g+\Delta g |^2) \mathcal{H}(f)\diag(| g+\Delta g |^2)- \diag(| g|^2) \mathcal{H}(f)\diag(| g |^2),\\
X_2&=\diag(\bar g) R^y\diag(g)-\diag(\overline{g+\Delta g}) R^y\diag(g+\Delta g), ~~~~~~~~~~~~~~~~~\\
X_3&=\diag(| g+\Delta g |^2) \mathcal{H}(\Delta f)\diag(| g+\Delta g |^2). ~~~~~~~~~~~~~~~~~~~~~~~~~~~~~~~
\end{align*}

To estimate the norm of the gradient, we derive the following lemma. 

\begin{lemma} \label{estimate} \begin{itemize} \item For matrices $X$ and $Y$,  $ \|\mathcal{H}^a(X+Y)\| \leq \| \mathcal{H}^a(X) \|+ \mathcal{H}^a(Y) \|$. 
\item Let the matrix $X=\diag(g_1) \mathcal{H}(f) \diag(g_2)$ for vectors $g_1, g_2,$ and $f$ in $\mathbb{C}^N$, we have $\| \mathcal{H}^a(X)\| \leq 2 \|f\| \|g_1\| \|g_2\|.$ 
\item  For vectors $g_1, g_2 \in \mathbb{C}^N$, $\| g_1 g_2\| <\|g_1\| \|g_2\|.$
\end{itemize}
\end{lemma}
By computation, 
\begin{align*}
X_1=\diag(|g+\Delta g|^2) \mathcal{H}(f)\diag(g \overline{\Delta g})+ \diag(|g+\Delta g|^2) \mathcal{H}(f)\diag(\Delta g \overline{g+\Delta g}) \\+\diag(g \overline{\Delta g}) \mathcal{H}(f) \diag(|g|^2)+\diag(\Delta g(\overline{g+ \Delta g}))\mathcal{H}(f)\diag(|g|^2), \\
X_2=-\diag( \bar g)R^y\diag(\Delta g)-\diag(\overline{\Delta g})R^y\diag(g+\Delta g), ~~~~~~~~~~~~~~~~~~~~~~~~~~~~\\
X_3=\diag(|g+\Delta g|^2)\mathcal{H}(\Delta f)\diag(|g+\Delta g|^2).~~~~~~~~~~~~~~~~~~~~~~~~~~~~~~~~~~~~~~~~~~
\end{align*}

Using Lemma \ref{estimate}, we conclude that 
\begin{align*}
 \|\mathcal{H}^a(X_1) \| &\leq 2 \|f\|\| \Delta g\|(\| |g+\Delta g|^2 \| \|g\|+ \| |g+\Delta g|^2 \| \| g+\Delta g \| +\| |g|^2\| \|g+\Delta g\|), \\
 \|\mathcal{H}^a(X_2) \| &\leq  \|\Delta g\| (\|g\|+\|g+\Delta g\|) \|g_{true}\|^2 \|f_{true}\| ,\\
  \|\mathcal{H}^a(X_3) \| &\leq  \|\Delta f\| \| |g+\Delta g|^2\|^2. 
 \end{align*}
}

\bibliographystyle{plain}	
\bibliography{myref}		
 
\end{document}